\documentclass[preprint,11pt,authoryear]{elsarticle}

\usepackage{amsmath,amssymb,amsthm}
\usepackage{graphicx}
\usepackage{booktabs}
\usepackage{bm}
\usepackage{siunitx}
\usepackage[hidelinks]{hyperref}
\usepackage[margin=1in]{geometry}
\usepackage{setspace}

\journal{Journal of Econometrics}

\makeatletter
\long\def\pprintMaketitle{\clearpage
  \iflongmktitle\if@twocolumn\let\columnwidth=\textwidth\fi\fi
  \resetTitleCounters
  \def\baselinestretch{1}%
  \printFirstPageNotes
  \begin{\elsarticletitlealign}%
 \thispagestyle{pprintTitle}%
   \def\baselinestretch{1}%
    \LARGE\@title\par\vskip18pt%
    \ifx\@elsarticlenewpageafter\newpage@after@title%
      \newpage
    \fi%
    \ifdoubleblind
      \vspace*{2pc}
    \else
      \normalsize\elsauthors\par\vskip10pt
      \footnotesize\itshape\elsaddress\par\vskip36pt
    \fi
    \ifx\@elsarticlenewpageafter\newpage@after@author%
      \newpage
    \fi%
    \ifvoid\absbox\else\unvbox\absbox\par\vskip20pt\fi
    \ifvoid\keybox\else\unvbox\keybox\par\vskip10pt\fi
    \end{\elsarticletitlealign}%
    \gdef\thefootnote{\arabic{footnote}}%
    \clearpage
  }

\renewenvironment{abstract}{\global\setbox\absbox=\vbox\bgroup
  \hsize=\textwidth\def\baselinestretch{1}%
  {\centering\textbf{\@elsarticleabstitle}\par}%
  \medskip\noindent\unskip\ignorespaces}
 {\egroup}

\def\keyword{%
  \def\sep{\unskip, }%
  \def\MSC{\@ifnextchar[{\@MSC}{\@MSC[2000]}}%
  \def\@MSC[##1]{\par\leavevmode\hbox{\it ##1~MSC:\space}}%
  \def\PACS{\par\leavevmode\hbox{\it PACS:\space}}%
  \def\JEL{\par\addvspace{8pt}\leavevmode\hbox{\it JEL:\space}}%
  \global\setbox\keybox=\vbox\bgroup\hsize=\textwidth
  \setstretch{1}\normalsize\normalfont
  \leftskip=2em\rightskip=2em\parskip\z@\parindent\z@
  \noindent\textit{\@elsarticlekwdtitle\@elsarticlekeywordtitlesep}%
  \ignorespaces}
\makeatother

\newtheorem{assumption}{Assumption}
\newtheorem{lemma}{Lemma}
\newtheorem{proposition}{Proposition}
\newtheorem{corollary}{Corollary}
\theoremstyle{definition}
\newtheorem{remark}{Remark}

\newcommand{\E}{\mathbb{E}}
\newcommand{\Prob}{\mathbb{P}}
\newcommand{\Cov}{\operatorname{Cov}}
\newcommand{\Var}{\operatorname{Var}}
\newcommand{\Ztil}{\tilde{Z}}
\newcommand{\indep}{\perp\!\!\!\perp}
\newcommand{\late}{\mathrm{LATE}}
\newcommand{\convd}{\overset{d}{\longrightarrow}}
\newcommand{\convp}{\overset{p}{\longrightarrow}}

\begin{document}
\setstretch{1.15}
\setlength{\emergencystretch}{2em}

\begin{frontmatter}

\title{What No First Stage Can Detect: Functional-Form Contamination in Linear IV}

\author{Parush Arora\fnref{aff}}
\fntext[aff]{Department of Economics, Ashoka University, Sonipat, India.}

\begin{abstract}
\begin{spacing}{1}%
\leftskip=2em\rightskip=2em\noindent
Applied instrumental variables (IV) practice reports a first-stage $F$, now often the
conditional $F$ of Sanderson and Windmeijer (2016), and reads a large value as license to
interpret the second stage. We show that no first-stage diagnostic can provide it. With a
scalar instrument, a scalar treatment, and covariates entered linearly, the 2SLS estimand
splits into a signal that a saturated specification would target and a contamination, the
covariance between curvature in the instrument propensity and a covariate level function.
The same nuisance sits in both terms, so it biases the estimand and inflates the reported
strength at once. When the instrument is nearly collinear with the covariates the signal
vanishes and the strength is manufactured. When the strength is honest the curvature still
biases the estimand through the outcome, where no first-stage number reaches it. We prove
that no functional of the joint distribution of instrument, treatment, and covariates can
detect this second bias, and we give a directed test built from reduced-form regressions, a
corrected estimator, and a reportable contamination share. Two applications show the modes.
A husband's insurance instrument (Olson, 1998) has a conditional $F$ above $36{,}000$, yet
correcting a linear income control more than doubles the estimate. The instrument of Nunn
and Wantchekon (2011) loses most of its first stage once geography enters flexibly.
\par\end{spacing}
\end{abstract}


\end{frontmatter}

%=============================================================================
\section{Introduction}\label{sec:intro}

The instrumental-variable (IV) estimator is the workhorse of applied microeconomics
whenever treatment is suspected to be endogenous. Its credibility rests on two
pillars, the exclusion restriction and the strength of the first stage. The first is
fundamentally untestable and is defended by argument and by sensitivity analysis
\citep{conley2012,kippersluis2018}. The second is treated as a statistical object to
be certified. A large literature, including \citet{bound1995}, \citet{staiger1997},
\citet{stock2005}, \citet{montielolea2013}, and \citet{andrews2019}, has produced a
battery of rule-of-thumb and size-based procedures for diagnosing weak instruments,
and reporting a first-stage $F$ is now routine, indeed mandatory in most journals.
When the model carries a rich set of controls, \citet{sanderson2016} show that the
relevant object is the \emph{conditional} first-stage $F$, computed after partialling
the covariates out of the excluded instrument, and applied work increasingly reports
it. Whether a large value of it should reassure is a separate question, and \citet{stevenson2026} argues on behavioural grounds that it need not, a first-stage $F$ above $200$ coexisting with substantial bias once researcher discretion enters. We give a mechanical counterpart that binds even for a single disciplined analyst running the standard specification.

Running in parallel, a second literature has clarified \emph{what} a valid, strong
instrument identifies. Since \citet{imbens1994} and \citet{angrist1996}, IV is
understood to recover a local average treatment effect (LATE) under treatment-effect
heterogeneity. \citet{heckman2005} formalize essential heterogeneity,
\citet{dechaisemartin2017} treats failures of monotonicity, and
\citet{mogstad2018,mogstad2024} synthesize the modern picture. Most consequential for
practice is \citet{blandhol2022}. Once covariates enter the second stage
\emph{linearly}, the near-universal specification, 2SLS is in general \emph{not} a
convex average of LATEs unless the covariates are entered nonparametrically.
\citet{sloczynski2022} treats the just-identified case, \citet{goldsmith2024} give a
general treatment of the contamination bias additive covariate adjustment induces
across \emph{multiple} treatments, and \citet{ddl2025} show that consistent interacted
2SLS requires the product of the IV propensity score and the covariates to be linear
in the covariates. It is by now understood that linearity of the instrument propensity
$\E[Z\mid X]$ in the covariates is the pivotal condition, and \citet{blandhol2022}
recommend assessing it directly with a Ramsey RESET on $\E[Z\mid X]$, a diagnostic on
$(Z,X)$ alone. \citet{ssu2026} synthesize these diagnostics.

\paragraph{The object of study.} We study linear IV with a scalar instrument $Z$, a
scalar treatment $D$, and covariates entered linearly, under conditional-on-$X$
validity. The population 2SLS coefficient on $D$ is the Frisch--Waugh--Lovell ratio
$\beta_{2SLS}=\E[\Ztil Y]/\E[\Ztil D]$, where $\Ztil=Z-L(Z\mid1,X)$ is the
covariate-residualized instrument. Writing the residual as $\Ztil=(Z-e(X))+h(X)$, with
$e(x)=\E[Z\mid X=x]$ the instrument propensity and $h=e-\ell$ its component orthogonal
to the linear covariate span, delivers an exact decomposition,
\begin{equation}\label{eq:introdecomp}
\beta_{2SLS}
=\frac{\overbrace{\E[\Cov(Z,Y\mid X)]}^{\text{signal }S_Y}+\overbrace{\Cov(h(X),m_Y(X))}^{\text{contamination }\theta_Y}}
      {\underbrace{\E[\Cov(Z,D\mid X)]}_{\text{signal }S_D}+\underbrace{\Cov(h(X),m_D(X))}_{\text{contamination }\theta_D}},
\end{equation}
where $m_D(x)=\E[D\mid X=x]$ and $m_Y(x)=\E[Y\mid X=x]$. The signal is the average within-cell covariance of the instrument with the outcome and with the treatment, the estimand a fully interacted (saturated) specification would target. Each
contamination term is a covariance between the propensity curvature $h$ and a covariate
level function, and so vanishes exactly when $h$ is orthogonal to that level function,
which a propensity linear in $X$ ($h\equiv0$) or a linear level function delivers as a
sufficient case. The binary-treatment LATE model is the leading special case, where
$\Cov(Z,D\mid X)=e(X)\{1-e(X)\}p(X)$ with $p$ the conditional first stage and the ratio
$S_Y/S_D$ is a convex average of conditional LATEs. The contamination, the
strength--bias identity, the impossibility theorem, the directed test, the correction,
and the reportable share developed below are properties of \emph{linear IV with a linear
covariate specification}, and none of them turns on the support of $Z$ or $D$.

\paragraph{What the general target means.} The bias throughout is measured against
$\bar L$, which is a convex average of causal effects only in the binary case. The
quantity of interest is never $\bar L$ in isolation but the gap between two estimators an
analyst chooses between. The estimand $\beta_{2SLS}$ is
what applied work reports, and $\bar L$ is the estimand of the saturated specification
that \citet{blandhol2022} recommend as the fix, so $\beta_{2SLS}-\bar L$ is the price of
the linear covariate shortcut relative to the recommended alternative, well defined and
consequential whatever the support of $D$. The binary case adds a causal reading, since
there $\bar L$ is a convex average of conditional LATEs and the estimand can be shown to
leave the LATE hull, which is why we run that case in full. Outside it $\bar L$ is the
saturated-propensity linear-IV estimand rather than a causal average, and the results
read as detecting functional-form contamination of the reported estimand relative to
that benchmark. The impossibility theorem, the directed test, the correction, and the
share concern this gap and hold whatever $\bar L$ means.

\paragraph{Two questions.} Two questions this literature leaves open motivate the
paper, and they are the questions a practitioner actually faces. First, \emph{what does
the strength diagnostic tell you about all this?} The implicit bridge in applied work,
that a strong (conditional) first stage licenses a causal reading of the second stage,
has never been examined. We show it is false, and more sharply that \emph{no}
first-stage diagnostic can repair it. Second, \emph{is checking the propensity enough?}
A propensity-only check such as a RESET on $\E[Z\mid X]$ cannot distinguish a biased
design from an unbiased one with equally nonlinear propensity, because whether the
nonlinearity matters depends on the outcome. These two negative results, and the
constructive test, correction, and share that answer them, are what is new here.

Two concurrent papers are the closest antecedents. In independent work
\citet{krumme2026} reach a related conclusion, that a diagnostic
must be directed at the outcome rather than the propensity alone. Their procedure
extends testable-inequality implications for IV validity to a joint null of instrument
validity and correct specification, and it requires a strong first stage. Ours instead
conditions on validity, isolates the contamination as a single conditional-moment
restriction, ties that restriction to the first-stage strength through the identity
below, and remains valid under weak identification. \citet{ssu2026} synthesize the
strength and heterogeneity diagnostics a practitioner already reports. Our object is the
one restriction those diagnostics cannot see, which we make precise as the impossibility
result of Section~\ref{sec:collin}.

\paragraph{Relation to the literature.} The results build on three lines of work.
\citet{blandhol2022} show that linear covariates break the LATE reading unless the
propensity is linear in $X$, and recommend saturating the controls, and we take that
level dependence as the starting point. Three things build on it. The first is the
strength--bias identity, that the same propensity curvature which biases the estimand
also inflates the conditional $F$ the analyst reports as reassurance
(Proposition~\ref{prop:identity}), so the diagnostic co-moves with the bias rather than staying uninformative about it. The second is that the propensity check those authors recommend
flags curvature that is merely present, so a bias-targeting test is needed to separate
the curvature that reaches the estimate from the curvature that does not. The third is a
size-controlled test valid under weak identification, a correction whose identification
cost is known in advance, and a reportable share. \citet{sloczynski2022} studies the
sign of the LATE weights in the just-identified design, a question about the causal
reading of $\bar L$ itself, whereas our object is the gap $\beta_{2SLS}-\bar L$ between
the reported estimand and the saturated benchmark and whether any first-stage number can
see it. \citet{goldsmith2024} study contamination across \emph{several} treatments, a
channel that survives a saturated propensity, whereas ours is single-treatment
functional form and vanishes under saturation, so the two are distinct and compound in a
design with both. The new primitives are the strength--bias identity and the reportable
share, and the impossibility theorem follows from them.

\paragraph{The identity, and two failure modes.} Our organizing observation is a
feature of \eqref{eq:introdecomp} that has gone unremarked. The \emph{denominator}
$S_D+\theta_D=\E[\Ztil D]$ is exactly the first-stage covariance the conditional $F$ is
built from, so the same nuisance $\theta_D$ that biases the estimand inflates the
reported strength (Proposition~\ref{prop:identity}). Because the contamination sits in
both numerator and denominator, it produces two distinct failure modes. In
\textbf{mode~(a)}, \emph{manufactured strength}, the instrument is nearly collinear with
the covariates, the signal $S_D$ vanishes, $\theta_D$ dominates the denominator, and the
conditional $F$ is inflated by the very curvature that biases the estimand
(Section~\ref{sec:collin}). In \textbf{mode~(b)}, \emph{invisible outcome-side bias}, the
strength is honest, $\theta_D$ is small and the $F$ large, yet the estimand is
biased through the numerator term $\theta_Y$, which loads on the outcome where no
first-stage number can reach it.

\paragraph{Results.} \emph{First}, Proposition~\ref{prop:collin} characterizes
collinearity between instrument and covariates, the pure form of mode~(a). Along any sequence in
which $\Var(Z\mid X)\to0$, the signal vanishes while the contamination does not, and the
estimand converges to $\Cov(h^\ast,m_Y^\ast)/\Cov(h^\ast,m_D^\ast)$, a ratio of covariate
contrasts with no causal content. \emph{Second}, Corollary~\ref{cor:Ffails} shows the
conditional first-stage slope stays bounded away from zero along the same sequence, so
the conditional $F$ \emph{diverges} in $n$ rather than collapsing, and
Proposition~\ref{prop:impossible} sharpens this into an impossibility result: no
functional of the law of $(Z,D,X)$, the conditional $F$, effective $F$, and
Cragg--Donald statistic in particular, can detect the bias, because it loads on the
conditional law of the outcome. \emph{Third}, we give a \emph{directed}
conditional-moment test of the null of no contamination, built from reduced-form
regressions alone (Section~\ref{sec:test}). It has power exactly where the $F$ is silent,
and its size is controlled even under weak identification, in both the screening form
and a signal-cleared bias-targeting form whose validity holds uniformly across the
collinearity range, where a saturate and compare contrast collapses. \emph{Fourth}, the same identity yields a
bias-corrected estimator with an exact no-free-lunch cost, the identifying variation it
surrenders equalling the manufactured share, and a reportable \emph{contamination share}
$\varrho$ to be read beside the conditional $F$ (Section~\ref{sec:correction}). The share
measures the first-stage channel alone, and like the $F$ it stays silent about the
outcome-side bias, so a small $\varrho$ narrows but does not certify a design, a point
the health insurance application makes concrete.

Each result is stated for the general model, with the binary LATE case noted as it
arises. The general model is a scalar treatment of arbitrary support. A vector of treatments, where the contamination studied here compounds with the cross-treatment channel of \citet{goldsmith2024}, is the one boundary we leave to future work. Section~\ref{sec:mc} confirms them in Monte Carlo designs spanning a multivalued
(ordered) treatment, a continuous instrument with a continuous treatment, and an
over-identified design in which the screen stacks moments across instruments. In each, as $\Var(Z\mid X)\to0$ the estimand leaves the hull of the
within-cell IV slopes while the conditional $F$ diverges, and the directed test fires.
Section~\ref{sec:empirics} works one application in full and adds two further designs.
The flagship, the effect of spousal health insurance coverage on wives' labour supply
\citep{olson1998}, is a clean instance of mode~(b). A
conditional $F$ above $36{,}000$ and a contamination share of $0.09$ are both
reassuring, yet entering husband income linearly, the standard choice, more than
doubles the estimated effect once corrected, and the corrected estimator stays fully
identified. The slave trade instrument of \citet{nunn2011}, with a continuous instrument
and a continuous treatment, is the real data face of mode~(a), where correcting the
propensity collapses a first stage that looked powerful, and the 401(k) eligibility
design corroborates the mode~(b) pattern in the most studied strong-instrument
application in the literature.

Section~\ref{sec:framework} sets up the general model and proves the decomposition and
the strength--bias identity. Section~\ref{sec:collin} develops the collinearity limit,
the failure of the conditional $F$, and the impossibility theorem.
Section~\ref{sec:test} constructs the test. Section~\ref{sec:correction} presents the
bias-corrected estimator and the identification trade-off. Section~\ref{sec:mc} reports
the Monte Carlo evidence and Section~\ref{sec:empirics} the applications.
Section~\ref{sec:practice} gives recommendations for practice and
Section~\ref{sec:conclusion} concludes. Proofs are collected in the Appendix.
%=============================================================================
\section{Framework and the strength--bias identity}\label{sec:framework}
%=============================================================================

Let $Z\in\mathbb{R}$ be a scalar instrument, $D\in\mathbb{R}$ a scalar treatment,
$Y\in\mathbb{R}$ an outcome, and $X\in\mathbb{R}^{d_x}$ a vector of covariates. The
decomposition that follows uses only the two maintained assumptions below, a linear
specification and regularity, and no structure on treatment assignment, so in
particular no binariness or monotonicity. For the \emph{interpretation} of the signal
ratio as a convex average of local average treatment effects, and hence for the
binary-treatment specialization we return to throughout, we additionally invoke the
standard conditional-on-$X$ LATE assumptions. We state them first, so that later
results may reference them, but they are used only for that interpretation. In the
LATE model potential treatments are $D(0),D(1)$ and potential outcomes $Y(d,z)$, with
observed $D=D(Z)$ and $Y=Y(D,Z)$.

\begin{assumption}[Conditional validity]\label{as:valid}
$\{Y(d,z),D(z)\}_{d,z}\indep Z \mid X$, and $Y(d,1)=Y(d,0)\equiv Y(d)$
(exclusion), almost surely.
\end{assumption}

\begin{assumption}[Conditional monotonicity]\label{as:mono}
$D_i(1)\ge D_i(0)$ almost surely.
\end{assumption}

\begin{assumption}[Conditional relevance]\label{as:rel}
The within-cell covariance of instrument and treatment $\Cov(Z,D\mid X=x)$ is nonzero on a
set of $x$ of positive probability, and its average $S_D=\E[\Cov(Z,D\mid X)]$ is
nonzero, so that the saturated specification benchmark $\bar L=S_Y/S_D$ is well
defined. In the binary specialization of Remark~\ref{rem:binary} this covariance is
$e(X)\{1-e(X)\}p(X)$ with $p(x)=\E[D\mid Z=1,X=x]-\E[D\mid Z=0,X=x]$, and instrument
overlap $0<e(x)<1$ together with $p(x)>0$ on a positive-probability set makes $S_D$
strictly positive. Relevance confined to strata where $\Var(Z\mid X)=0$ would leave
$S_D=0$, and the estimand, though still defined through $\E[\Ztil D]=\theta_D$ whenever
$\theta_D\neq0$, would be pure contamination with no $\bar L$ to benchmark against.
\end{assumption}

\begin{assumption}[Maintained linear specification]\label{as:lin}
The analyst estimates the just-identified linear IV model
$Y=\beta D + X'\gamma + \varepsilon$, using $Z$ as the single excluded
instrument and $(1,X')$ as included instruments.
\end{assumption}

\begin{assumption}[Regularity]\label{as:reg}
$Y$, $D$, $Z$ and $X$ have finite second moments, and $\E[(1,X')'(1,X')]$ is
nonsingular. In addition, $\E[\Ztil D]\neq 0$, where $\Ztil\equiv Z-L(Z\mid 1,X)$ and
$L(\cdot\mid 1,X)$ denotes the population linear projection on $(1,X')$.
\end{assumption}

Under Assumption~\ref{as:lin}, the population 2SLS coefficient on $D$ admits the
familiar FWL / partialling-out form
\begin{equation}\label{eq:fwl}
\beta_{2SLS}=\frac{\E[\Ztil\,Y]}{\E[\Ztil\,D]} .
\end{equation}
Define the instrument propensity $e(x)=\Prob(Z=1\mid X=x)=\E[Z\mid X=x]$, its
linear projection $\ell(x)=L(Z\mid 1,X)(x)$, and the \emph{propensity curvature}
\begin{equation}\label{eq:h}
h(x)\;\equiv\; e(x)-\ell(x),\qquad \E[h(X)]=0,\quad \E[h(X)\,r(X)]=0
\ \text{for all affine } r.
\end{equation}
Let $m_D(x)=\E[D\mid X=x]$, $m_Y(x)=\E[Y\mid X=x]$, and let
$\late(x)=\E[Y(1)-Y(0)\mid D(1)>D(0),X=x]$ denote the conditional LATE.

\begin{lemma}[Contamination decomposition]\label{lem:decomp}
Under Assumptions~\ref{as:lin}--\ref{as:reg},
\begin{equation}\label{eq:decomp}
\beta_{2SLS}=\frac{S_Y+\theta_Y}{S_D+\theta_D},
\end{equation}
where the \emph{signal} terms are the averaged within-cell covariances
\[
S_Y=\E\!\big[\Cov(Z,Y\mid X)\big],\qquad
S_D=\E\!\big[\Cov(Z,D\mid X)\big],
\]
and the \emph{contamination} terms are the covariances of the propensity curvature
with the covariate level functions,
\[
\theta_Y=\Cov\!\big(h(X),m_Y(X)\big)=\E[h(X)Y],\quad
\theta_D=\Cov\!\big(h(X),m_D(X)\big)=\E[h(X)D].
\]
Consequently, writing $\bar L\equiv S_Y/S_D$ for the saturated specification estimand,
\begin{equation}\label{eq:bias}
\beta_{2SLS}-\bar L=\frac{\theta_Y-\bar L\,\theta_D}{S_D+\theta_D},
\qquad\text{with } S_D+\theta_D=\E[\Ztil D].
\end{equation}
\end{lemma}

The signal $S_Y/S_D$ is the estimand of a fully interacted (saturated)
specification. Writing $\tau(x)=\Cov(Z,Y\mid X=x)/\Cov(Z,D\mid X=x)$ for the
within-cell IV slope, $\bar L=\E[\Cov(Z,D\mid X)\,\tau(X)]/\E[\Cov(Z,D\mid X)]$ is a
weighted average of within-cell slopes with weights proportional to the within-cell
covariance of instrument and treatment. We say the \emph{sign condition} holds when
$\Cov(Z,D\mid X=x)$ has a single sign across $x$. Under it these weights are
non-negative, $\bar L$ is a genuine convex average, and it lies in the hull of the
within-cell slopes, the reading we use in the Monte Carlo of Section~\ref{sec:mc} and
in the applications of Section~\ref{sec:empirics}. Without it the weights can change
sign and $\bar L$ need not lie in that hull. In the binary specialization of
Remark~\ref{rem:binary} monotonicity (Assumption~\ref{as:mono}) delivers the sign
condition, and for a general treatment it is an additional restriction we flag where it
is used. The contamination terms are covariances between the propensity
curvature $h$ and the covariate level functions. By \eqref{eq:h} they equal the
covariance between the \emph{nonlinear} part of the propensity and the
\emph{nonlinear} part of the level function, so each vanishes exactly under
orthogonality, $\theta_D=0\iff h\perp m_D$ and $\theta_Y=0\iff h\perp m_Y$. A linear
propensity ($h\equiv0$) kills both terms and reproduces \citet{blandhol2022}, while a
linear level function kills only its own. Orthogonality can also hold with neither
function linear.

\begin{remark}[Binary-treatment LATE specialization]\label{rem:binary}
Suppose $Z,D\in\{0,1\}$ and Assumptions~\ref{as:valid}--\ref{as:rel} hold. Then
$\Cov(Z,D\mid X)=e(X)\{1-e(X)\}p(X)$ and, by the conditional Wald identity,
$\Cov(Z,Y\mid X)=e(X)\{1-e(X)\}p(X)\late(X)$. Hence $S_D=\E[e(1-e)p]$,
$S_Y=\E[e(1-e)p\,\late]$, the within-cell slope $\tau(x)$ is the conditional LATE, and
$\bar L=S_Y/S_D$ is a convex average of conditional LATEs, the
saturated specification estimand of \citet{blandhol2022}. Monotonicity is what
single-signs $\Cov(Z,D\mid X)$ and so makes the average convex, and it plays no role in
the decomposition itself.
\end{remark}

The convex-average reading of $\bar L$ rests on the sign condition alone, and we collect the primitive cases in which that condition holds, so that it reads as a bounded
scope statement rather than a caveat invoked in passing.

\begin{remark}[Sufficient conditions for the sign condition]\label{rem:sign}
The sign condition, that $\Cov(Z,D\mid X=x)$ has a single sign across $x$, holds in each of
the following cases. (i) Binary $Z$ and binary $D$ under conditional monotonicity
(Assumption~\ref{as:mono}), where $\Cov(Z,D\mid X)=e(X)\{1-e(X)\}p(X)$ and overlap with
$p(x)\ge0$ fixes the sign, the case of Remark~\ref{rem:binary}. (ii) An ordered treatment
$D=\sum_j\mathbf 1\{D\ge j\}$ under stagewise monotonicity, where each threshold crossing is
monotone in $Z$ and $\Cov(Z,D\mid X)$ is a sum of nonnegative stage covariances, the
ordered design of Section~\ref{sec:mc}. (iii) A continuous $Z$ and $D$ with a within-cell
first stage of constant sign, $\partial_z\E[D\mid Z=z,X=x]$ of one sign in $(z,x)$, which
holds for a monotone structural first stage. Outside these cases the weights can change
sign and $\bar L$ need not lie in the hull of within-cell slopes. The condition governs the
\emph{interpretation} of $\bar L$ only. The decomposition of Lemma~\ref{lem:decomp}, the
strength--bias identity, the impossibility result, the directed test, and the correction
hold without it.
\end{remark}

The signal ratio $\bar L$ is the object a saturated specification targets, and we take
it as the benchmark. Our subject is a feature of the decomposition that has gone
unremarked, which we state as the paper's organizing result.

\begin{proposition}[Strength--bias identity]\label{prop:identity}
The denominator of the 2SLS estimand is the population first-stage covariance
$\E[\Ztil D]=S_D+\theta_D$ from which the Sanderson--Windmeijer conditional $F$ is
constructed. The contamination $\theta_D$ is the common term. It enters the
denominator of the bias $\beta_{2SLS}-\bar L=(\theta_Y-\bar L\theta_D)/(S_D+\theta_D)$
and, being a signed covariance, inflates the reported first-stage covariance when
$\theta_D>0$ and deflates it when $\theta_D<0$. The displacement of $\beta_{2SLS}$
from $\bar L$ requires the numerator $\theta_Y-\bar L\theta_D\neq0$, so strength and
interpretability need not move together. The point of the identity is the weaker but
consequential one that they \emph{can} be driven by the same nuisance $\theta_D$, and
under collinearity (Section~\ref{sec:collin}) they are.
\end{proposition}

\begin{proof}
This is immediate from Lemma~\ref{lem:decomp}. The denominator of \eqref{eq:decomp} is
$S_D+\theta_D$, and $\E[\Ztil D]=\E[(h(X)+(Z-e(X)))D]=\theta_D+S_D$ since
$\E[(Z-e(X))D\mid X]=\Cov(Z,D\mid X)$ integrates to $S_D$.
\end{proof}

\noindent The identity organizes the two failure modes. Through the
\emph{denominator} the contamination inflates the reported strength, dominant under
collinearity (mode~a, Section~\ref{sec:collin}). Through the \emph{numerator}, via
$\theta_Y$, it biases the estimand even when the reported strength is honest, a bias
that loads entirely on the outcome (mode~b, the health insurance design of
Section~\ref{sec:hi}).

\subsection*{The leading binary case in one pass}
Before developing the results for a treatment of arbitrary support, we run the whole
argument once in the binary-treatment LATE model of Remark~\ref{rem:binary}, where every
object has a causal reading. It is the case most applied designs occupy, and it fixes the
intuition that the general theory then extends.

Here the benchmark $\bar L=S_Y/S_D$ is a genuine convex average of conditional
LATEs, so a saturated specification recovers a causal parameter and the only question is
what the 2SLS estimate does to it. The bias is
$\beta_{2SLS}-\bar L=(\theta_Y-\bar L\theta_D)/(S_D+\theta_D)$, and both contamination
terms are covariances of the propensity curvature $h$ with a covariate level function.
The strength--bias identity (Proposition~\ref{prop:identity}) is the observation that the
denominator $S_D+\theta_D$ is exactly the first-stage covariance the conditional $F$ is
built from, so the nuisance $\theta_D$ that shifts the estimand is the same one that moves
the reported strength.

That single fact produces the two modes. In mode~a the instrument is nearly a nonlinear
function of the covariates, the signal $S_D$ vanishes, and $\theta_D$ manufactures the
conditional $F$ out of the very curvature that biases the estimand. In mode~b the strength
is honest, yet $\theta_Y$ displaces the estimand through the outcome, where no first-stage
number can reach, so here the estimand can leave the hull of conditional LATEs while the
$F$ stays large (Section~\ref{sec:mcbinary}). The impossibility result, the directed test,
and the correction previewed above all hold in this binary model with a causal reading of
$\bar L$, and the remainder of the paper establishes them for a scalar treatment of
arbitrary support, with the binary case as the special instance in which $\bar L$ is a
convex average of LATEs.

%=============================================================================
\section{Collinearity and the failure of the first-stage diagnostic}\label{sec:collin}
%=============================================================================

We now formalize what happens as the instrument becomes collinear with the
covariates. Intuitively, ``$Z$ is collinear with $X$'' means $Z$ is close to a
(possibly nonlinear) function of $X$, i.e.\ the conditional variance
$\Var(Z\mid X)$ is small. In the binary case this variance equals $e(X)\{1-e(X)\}$. We
index a sequence of data-generating processes (DGPs) by $t$.

\begin{assumption}[Collinearity sequence]\label{as:seq}
Along $t=1,2,\dots$ the following hold. (i) $\Var_t(Z\mid X)\to 0$ for almost every
$x$ and $\E[\Var_t(Z\mid X)]\to 0$. (ii) The propensity curvature converges,
$h_t\to h^\ast$ in $L^2$, with $h^\ast$ non-degenerate. (iii) The within-cell slope
and the level functions converge in $L^2$ to limits $\tau^\ast,m_D^\ast,m_Y^\ast$,
with $\sup_t\E[\Var_t(D\mid X)]<\infty$ and $\{\Var_t(Y\mid X)\}$ uniformly
integrable. (iv) $\Cov(h^\ast,m_D^\ast)\neq 0$.
\end{assumption}

Part (i) is the collinearity limit. Part (ii) says the propensity retains curvature as
$Z$ is pushed toward a deterministic function of $X$, and a nonlinear index is the
leading example, of which Section~\ref{sec:mc} uses several. Part (iv) is the genericity
condition that makes the limit interesting. It holds whenever the nonlinear part of
$m_D=\E[D\mid X]$ is not orthogonal to $h^\ast$, and it rules out a knife-edge
cancellation of the propensity curvature by curvature in the baseline conditional mean
rather than following from relevance alone.

\begin{proposition}[Collinearity contamination limit]\label{prop:collin}
Under Assumptions~\ref{as:lin}--\ref{as:reg} and \ref{as:seq}, the signal terms
vanish, $S_{D,t}\to 0$ and $S_{Y,t}\to 0$, while the contamination terms converge to
$\Cov(h^\ast,m_D^\ast)\neq0$ and $\Cov(h^\ast,m_Y^\ast)$. Hence
\begin{equation}\label{eq:limit}
\beta_{2SLS,t}\;\longrightarrow\;
\frac{\Cov(h^\ast,m_Y^\ast)}{\Cov(h^\ast,m_D^\ast)}.
\end{equation}
The limit is a ratio of covariate-level contrasts and carries no causal content. In
the binary LATE specialization of Remark~\ref{rem:binary} it may lie outside
$[\operatorname*{ess\,inf}_x \late^\ast(x),\ \operatorname*{ess\,sup}_x \late^\ast(x)]$
and take the opposite sign of every $\late^\ast(x)$.
\end{proposition}

The reduced-form intuition is transparent. As $\Var(Z\mid X)\to0$, the
within-cell experimental variation in the instrument disappears. All
that survives in $\Ztil=Z-\ell(X)$ is the projection error $h(X)$, which is a
deterministic function of $X$. The estimator then compares $Y$ and $D$ across
covariate cells using weights $h(X)$, exactly the comparison an IV design is
meant to avoid. The next result is the payoff for practice.

\begin{corollary}[The conditional $F$ does not warn]\label{cor:Ffails}
The population Sanderson--Windmeijer conditional first-stage slope is
$\pi_t=\E[\Ztil_t D]/\Var(\Ztil_t)=(S_{D,t}+\theta_{D,t})/\Var(\Ztil_t)$, and
under Assumption~\ref{as:seq},
\[
\pi_t\;\longrightarrow\;\frac{\Cov(h^\ast,m_D^\ast)}{\E[(h^\ast)^2]}\neq 0 .
\]
Consequently the conditional first-stage $F$ statistic, whose population
counterpart is proportional to $n\,\pi_t^2\,\Var(\Ztil_t)/\sigma^2_{v}$, is
bounded away from zero and diverges with the sample size along the sequence,
even though by Proposition~\ref{prop:collin} the estimand is pure contamination.
A large conditional $F$ therefore places no bound on the contamination bias
\eqref{eq:bias}.
\end{corollary}

Corollary~\ref{cor:Ffails} is the formal statement of the applied warning. The
conditional $F$ of \citet{sanderson2016} correctly certifies that $\Ztil$ is a
strong predictor of $D$. Under collinearity, however, $\Ztil$ predicts $D$
\emph{through the contamination channel} $h(X)$, not through the within-cell signal
$\Cov(Z,D\mid X)$. Strength and interpretability are distinct, and can even be driven by the same nuisance.

Two qualifications keep the warning honest and locate where it bites. First, the
collinearity that drives mode~(a) is itself visible without the outcome. As
$\Var(Z\mid X)\to0$ the instrument is nearly a function of the covariates, so a
regression of $Z$ on a flexible basis in $X$ returns a fit approaching one, and the
within-cell signal announces itself as weak even while the conditional $F$ is large. A
diligent analyst can therefore see manufactured strength coming in mode~(a) from the
fit of the instrument on the covariates, and the sharp content of Corollary~\ref{cor:Ffails} is that
the $F$ points the wrong way exactly when that fit is what should be checked. Second, and
by consequence, the impossibility bites hardest in mode~(b), where the signal is honest,
the fit of the instrument on the covariates is unremarkable, and only the outcome reveals the bias.
The bias-targeting test of Section~\ref{sec:test} inherits this division. Its power
against a fixed alternative is full when the signal is bounded away from zero, the
mode~(b) regime, and declines as the signal vanishes, the mode~(a) regime, which is the
Anderson--Rubin trade-off of Proposition~\ref{prop:tradeoff}. The two do not leave a gap,
because the regime in which the test loses power is the regime the fit of the instrument on the covariates already flags.

The following result shows the failure is not special to the conditional $F$ but
is shared by \emph{every} strength diagnostic, and explains why. Such diagnostics
are computed from the first stage, whereas the bias lives in the outcome.

\begin{proposition}[Strength diagnostics cannot detect contamination]
\label{prop:impossible}
Let $T$ be any diagnostic that is a measurable functional of the joint
distribution of $(Z,D,X)$, in particular any first-stage relevance or strength
statistic, including the conditional $F$, the effective $F$ of
\citet{montielolea2013}, and the Cragg--Donald statistic. Fix any such law with
nonzero curvature, $h\neq0$. Then $T$ is uninformative about the contamination
bias. There exist data-generating processes with that common joint distribution of
$(Z,D,X)$, hence identical values of $T$, but arbitrarily different biases
$\beta_{2SLS}-\bar L$. Detecting the bias requires the conditional distribution of
$Y$. The restriction to $h\neq0$ is not a weakening. When $h\equiv0$ the first stage
already certifies $\theta_D=\theta_Y=0$ and hence no contamination, by
Lemma~\ref{lem:decomp}, so the interesting case is exactly the one $T$ cannot
resolve.
\end{proposition}

Proposition~\ref{prop:impossible} is the deep reason the applied bridge fails. No
refinement of the first-stage report, however sophisticated, can be made to carry
information about the contamination. By \eqref{eq:bias} the bias depends on
$\theta_Y=\E[h(X)Y]$ and on $\bar L=S_Y/S_D$, both functionals of the outcome that
are left completely free once the law of $(Z,D,X)$ is fixed. This is an instance
of the \emph{level dependence} of \citet{blandhol2022}. Concretely, the free
direction is the outcome level of the units the instrument does not move. Shifting the
conditional outcome mean $m_Y$ in the curvature direction $h$ changes $\theta_Y$, and
hence the bias, while leaving the law of $(Z,D,X)$, the within-cell slopes, and every
strength diagnostic unchanged. A strength statistic summarizes how the instrument moves
the treatment, whereas the bias lives with the level of the units it does not move, so a
valid diagnostic must use $Y$. One thing the first stage \emph{can} certify is the absence of
curvature: by Lemma~\ref{lem:decomp}, $h\equiv0$, a propensity linear in $X$ and
testable from $(Z,X)$ alone by a RESET, forces $\theta_D=\theta_Y=0$ and hence no
contamination. The impossibility is therefore conditional on nonzero curvature. Once
$h\neq0$, no first-stage functional can say whether that curvature is harmful,
because harm runs through $\theta_Y$ and $\bar L$, which only the outcome pins down.
This is the precise content of the blanket statements elsewhere that a strong first
stage cannot license a causal reading. This cuts in the opposite direction to the
familiar impossibility that instrument \emph{validity} cannot be confirmed even
with
infinite data \citep{ssu2026}. There the outcome is of no help, while here it is
exactly what a diagnostic requires. The observation both motivates and disciplines
the test of the next section, whose central moment $\theta_Y=\E[h(X)Y]$ is
precisely such an outcome functional. It also implies, a fortiori, that a Ramsey
RESET applied to the first stage, a functional of $(Z,X)$ alone, cannot
distinguish a biased design from an unbiased one. We confirm this in
Section~\ref{sec:mc}.

%=============================================================================
\section{A directed test for contamination}\label{sec:test}
%=============================================================================

\subsection{The null and its interpretation}

By Lemma~\ref{lem:decomp}, $\beta_{2SLS}=\bar L$ whenever $\theta_D=\theta_Y=0$.
We therefore test
\begin{equation}\label{eq:null}
H_0:\ \theta_D=0\ \text{ and }\ \theta_Y=0
\qquad\text{against}\qquad
H_1:\ (\theta_D,\theta_Y)\neq(0,0).
\end{equation}
Under $H_0$ the 2SLS estimand retains its interpretation as a convex average of LATEs. Because $\theta_D=\E[h(X)D]$ and $\theta_Y=\E[h(X)Y]$ depend only
on the propensity curvature and on observables, $H_0$ is a statement about the
\emph{reduced form}. The consequential misspecification is nonlinearity of the
instrument propensity $\E[Z\mid X]$ in the directions that also load on $D$ and
$Y$. This is what distinguishes the test from a Ramsey RESET applied to the first
stage $D\sim(Z,X)$. RESET rejects for propensity curvature in \emph{any}
direction, including curvature orthogonal to the outcome that leaves
\eqref{eq:bias} unaffected. This screening statistic instead probes only curvature
that loads on $D$ and $Y$, and does so without touching an endogenous regressor. It
is still a \emph{sufficient}-condition test. It can reject when the loadings satisfy
$\theta_Y=\bar L\theta_D$ and the estimand is in fact unbiased (harmless curvature),
which is why Section~\ref{sec:exact} adds the necessary-and-sufficient
bias-targeting variant that probes the single direction $Y-\bar L D$.

\subsection{Estimator and limiting distribution}

Let $b(X)=(b_1(X),\dots,b_K(X))'$ be a vector of basis functions (for example
polynomials or splines) orthogonalized against $(1,X')$ in the population, so that
$\E[b(X)(1,X')]=0$. These span the curvature directions to be probed. Model the
propensity as $e(X)=\alpha_0+X'\alpha_1+b(X)'\delta+\rho(X)$ with
$\E[\rho(X)b(X)]=0$, so that the component of the curvature captured by the basis
is $h_K(X)=b(X)'\delta$. Given an i.i.d.\ sample $\{(Y_i,D_i,Z_i,X_i)\}_{i=1}^n$,
the test is computed in three steps.

\begin{enumerate}
\item[(i)] Regress $Z_i$ on $(1,X_i',b(X_i)')$ by OLS and form
$\hat h_i=b(X_i)'\hat\delta$.
\item[(ii)] Form the sample contamination moments
$\hat\theta_D=n^{-1}\sum_i \hat h_i D_i$ and
$\hat\theta_Y=n^{-1}\sum_i \hat h_i Y_i$, stacked as
$\hat\theta=(\hat\theta_D,\hat\theta_Y)'$.
\item[(iii)] Compute the Wald statistic
$W_n=n\,\hat\theta'\hat\Sigma^{-1}\hat\theta$, where $\hat\Sigma$ is a consistent
estimator of the asymptotic variance $\Sigma$ of $\sqrt n\,\hat\theta$.
\end{enumerate}

\begin{assumption}[Test regularity]\label{as:test}
(i) $\{(Y_i,D_i,Z_i,X_i)\}_{i=1}^n$ are i.i.d.\ with $\E[Y^4]<\infty$,
$\E[D^4]<\infty$ and $\E\|b(X)\|^4<\infty$. (ii) $G=\E[b(X)b(X)']$ is nonsingular
and the basis dimension $K$ is fixed. (iii) $(\alpha_0,\alpha_1,\delta)$ are the
unique population least-squares coefficients in
$e(X)=\alpha_0+X'\alpha_1+b(X)'\delta+\rho(X)$ with
$\E[\rho(X)(1,X',b(X)')']=0$. (iv) $\Sigma=\Var(\psi_i)$ is nonsingular.
\end{assumption}

\begin{proposition}[Limiting distribution]\label{prop:test}
Under Assumption~\ref{as:test},
$\sqrt n\,(\hat\theta-\theta)\convd N(0,\Sigma)$, where $\Sigma=\Var(\psi_i)$ is the
asymptotic variance of the two-step moment estimator. The influence function $\psi_i$
is the stacked contamination moment $(b(X_i)'\delta\,W_i-\theta_W)_{W\in\{D,Y\}}$ plus a
generated-regressor correction for the first-step estimation of $\delta$, given
explicitly in the Appendix. Under $H_0$, $W_n\convd \chi^2_2$. Against a fixed
alternative with $(\theta_D,\theta_Y)\neq0$, $W_n\to\infty$, so the test is consistent.
Against local alternatives $\theta=\eta/\sqrt n$ it has nontrivial power with
noncentrality $\eta'\Sigma^{-1}\eta$.
\end{proposition}

The estimator $\hat\theta$ is a two-step GMM estimator, a first-step OLS for
$(\alpha_0,\alpha_1,\delta)$ followed by the second-step sample moments (ii).
Under Assumption~\ref{as:test}, \citet[Theorem 6.1]{newey1994} give asymptotic
linearity with the influence function above and $\Sigma=\Var(\psi_i)$. The second
term of $\psi_i$ is the generated-regressor correction for first-step estimation
of $\delta$, and it is nonzero precisely because $b(X)$ enters through
$\hat\delta$. The plug-in $\hat\Sigma=n^{-1}\sum_i\hat\psi_i\hat\psi_i'$, formed
from the sample analogues $\hat\psi_i$, is consistent by the law of large numbers
and the continuous mapping theorem, so $W_n\convd\chi^2_2$ under $H_0$. Because
$\hat\theta$ is a Hadamard-differentiable functional of the empirical measure, the
nonparametric pairs bootstrap is first-order valid by the functional delta method
\citep[Theorem 23.9]{vandervaart1998}. Resampling $(Y_i,D_i,Z_i,X_i)$ jointly
reproduces the generated-regressor term automatically, and we use it throughout
Section~\ref{sec:mc}.

The null \eqref{eq:null} is a statement about the reduced form, and the estimator
$\hat\theta$ never divides by the first-stage covariance $\E[\Ztil D]$. The test is
therefore immune to the weak-instrument problem, which is a small-denominator
problem for $\beta_{2SLS}$. This is what separates the diagnostic from any
procedure built on an IV estimate, and because the collinearity limit is precisely
where it matters, we state it formally rather than leave it to
Remark~\ref{rem:hausman}.

\begin{proposition}[Size under weak identification, power under manufactured strength]\label{prop:weakrobust}
Call $\{P_n\}$ a \emph{weak-identification null sequence} if the following hold. (i) The
fourth moments of $(Y,D,\|b(X)\|)$ are bounded uniformly in $n$ and
$\Sigma_n\to\Sigma_\ast\succ0$, so Assumption~\ref{as:test} holds uniformly along the
sequence. (ii) The first-stage strength drifts to zero, $S_{D,n}\to0$, as when
$\Var(Z\mid X)\to0$ or $p(X)\to0$. (iii) The curvature stays nondegenerate,
$h_n\to h^\ast\neq0$. (iv) The null is maintained, $\theta_{D,n}=\theta_{Y,n}=0$ for all
$n$. Along any such sequence, under $H_0$, $W_n\convd\chi^2_2$, so the test that rejects
when $W_n>\chi^2_{2,1-\alpha}$ has asymptotic size $\alpha$ regardless of the degree of
identification, and the same holds for the orthogonalized statistic $W^o_n$ of
Proposition~\ref{prop:orthogonal}. By
contrast, the contaminated collinearity sequence of Assumption~\ref{as:seq} has
$\Cov(h^\ast,m_D^\ast)\neq0$, hence $\theta_{D,n}\not\to0$, and is an \emph{alternative}
rather than a null. Along it the noncentrality $n\,\theta_n'\Sigma_n^{-1}\theta_n$
diverges and the test's power tends to one, even as every first-stage strength
diagnostic diverges with the sequence (Corollary~\ref{cor:Ffails}).
\end{proposition}

The mechanism is that $\theta=(\theta_D,\theta_Y)$ and its influence function
$\psi$ are functionals of the reduced form $(Z,D,Y,X)$, estimated at rate $\sqrt n$
by the regression of $Z$ on $(1,X',b(X)')$, whose behaviour is governed by
$G=\E[bb']$ rather than by the relevance of the instrument. Because the curvature
stays nondegenerate, $h_n\to h^\ast\neq0$, the leading influence term
$b(X)'\delta\,W-\theta_W$ stays nondegenerate and $\Sigma_n$ remains nonsingular
even as the generated-regressor correction shrinks with the residual variance of
$Z$. The screening test thus retains its size in exactly the regime, $S_D\to0$, in
which every first-stage strength diagnostic diverges (Corollary~\ref{cor:Ffails})
and a saturate and compare Hausman contrast collapses. The point is that weak
identification and contamination are separate. The null sequence has $S_D\to0$
without contamination, whereas Assumption~\ref{as:seq} couples $S_D\to0$ with
$\theta_D\not\to0$ and is precisely the design the test is meant to reject. The bias-targeting refinement of Section~\ref{sec:exact} does require
$S_D$ bounded away from zero, for the separate reason that it plugs in
$\hat{\bar L}$. That boundary is discussed there.

Proposition~\ref{prop:test} fixes the basis dimension $K$, so the curvature is a
parametric first step. Applied users will often prefer to estimate the propensity
flexibly or by machine learning, which breaks Assumption~\ref{as:test}(ii). The
first-step error is then no longer $\sqrt n$-negligible, and a naive plug-in is
biased at rate faster than $\sqrt n^{-1}$. The following result restores valid
inference by replacing the plug-in moment with a Neyman-orthogonal one, so that
only $n^{1/4}$-consistent, cross-fitted nuisances are required.

\begin{proposition}[Orthogonalized test with a data-driven propensity]\label{prop:orthogonal}
Let the nuisances be the propensity $e(\cdot)$ and the level functions
$m_W(\cdot)=\E[W\mid X]$, $W\in\{D,Y\}$, and write
$\tilde m_W(X)=m_W(X)-L(m_W\mid 1,X)$ for the level function detrended of its
affine part. For $W\in\{D,Y\}$ define the orthogonalized moment
\begin{equation}\label{eq:ortho}
\psi^{o}_{W}(O;e,m_W)=h(X)\,W+\tilde m_W(X)\,\big(Z-e(X)\big)-\theta_W,
\qquad h=e-L(e\mid 1,X),
\end{equation}
which has mean zero at the truth and is Neyman-orthogonal. Its pathwise derivative
with respect to $e$ and with respect to $m_W$ vanishes. Estimate $(e,m_D,m_Y)$ by
any method with $K$-fold cross-fitting, and form
$\hat\theta^{o}=(\hat\theta^{o}_D,\hat\theta^{o}_Y)'$,
$\hat\theta^{o}_W=n^{-1}\sum_i\{\hat h(X_i)W_i+\hat{\tilde m}_W(X_i)(Z_i-\hat e(X_i))\}$,
and $W^{o}_n=n\,\hat\theta^{o\prime}\hat\Sigma_o^{-1}\hat\theta^{o}$. Suppose the
cross-fitted nuisances satisfy $\|\hat e-e\|_{L^2}=o_p(n^{-1/4})$ and
$\|\hat m_W-m_W\|_{L^2}=o_p(n^{-1/4})$, hence the product rate
$\|\hat e-e\|_{L^2}\|\hat m_W-m_W\|_{L^2}=o_p(n^{-1/2})$, with the fourth moments
of Assumption~\ref{as:test}(i) and $\Sigma_o=\Var(\psi^o_i)$ nonsingular. Then
$\sqrt n(\hat\theta^{o}-\theta)\convd N(0,\Sigma_o)$ and, under $H_0$,
$W^{o}_n\convd\chi^2_2$.
\end{proposition}

The fixed-basis test of Proposition~\ref{prop:test} and this orthogonalized test are
two ends of a spectrum, with an undersmoothed fixed-form series ($K=o(n^{1/2})$ with the
two-step correction of \citet{newey1994}) as the intermediate route. It is
orthogonalization~\eqref{eq:ortho} that removes the $n^{1/4}$ bottleneck and licenses
off-the-shelf machine learning for the propensity. We use a fixed cubic basis throughout
and treat $K$ as a robustness dimension.

\begin{remark}[What a non-rejection certifies]\label{rem:power}
The fixed-basis screen of Proposition~\ref{prop:test} has power only against
contamination that projects onto the span of $b(X)$. Its noncentrality loads on
$\E[b(X)D]$ and $\E[b(X)Y]$, so any curvature $h$ orthogonal to $b$ leaves the
statistic unaffected. A non-rejection therefore certifies only that there is no
detectable contamination \emph{in the probed directions}, not that
$\theta_D=\theta_Y=0$ in every direction. This is the price of directing the test.
The growing-basis and orthogonalized versions
(Proposition~\ref{prop:orthogonal}) restore consistency against all directions as
$K\to\infty$. With the fixed cubic basis used below, the empirical ``clears'' of
Section~\ref{sec:empirics} should accordingly be read as ``no contamination
detectable at cubic order,'' and $K$ as a robustness dimension one can raise.
\end{remark}

\begin{remark}[Why not just saturate and compare?]\label{rem:hausman}
The natural alternative to our test is to estimate $\beta$ under a saturated
covariate specification and compare it to the linear-specification estimate via a
Hausman contrast. Under collinearity this is exactly the wrong tool. Saturating
drives the effective instrument toward the vanishing signal term $S_D\to0$
(Proposition~\ref{prop:collin}), so the saturated estimator is weakly identified
and its sampling variability swamps the contrast. Our test sidesteps this because
it never forms an IV ratio, and its size is controlled along the collinearity
sequence (Proposition~\ref{prop:weakrobust}), the property the saturate and compare
contrast cannot have.
\end{remark}

\subsection{A bias-targeting variant}\label{sec:exact}

The null \eqref{eq:null} is \emph{sufficient} for $\beta_{2SLS}=\bar L$ but not
necessary. By \eqref{eq:bias} the bias vanishes whenever
$\theta_Y=\bar L\,\theta_D$, which can hold with $(\theta_D,\theta_Y)\neq0$. The
leading example is \emph{harmless curvature}. Here the propensity is nonlinear in
$X$, so $h\neq0$ and, since $m_D$ inherits the curvature, $\theta_D\neq0$, yet the
outcome level is linear in $X$ and the effect is homogeneous, so
$\theta_Y=\bar L\theta_D$ and the estimand is unbiased. Testing \eqref{eq:null}
then rejects a design that is in fact fine. The screening test of
Section~\ref{sec:test} is therefore best read as a strongly identified,
conservative first pass.

The exact no-bias restriction is a single moment. Since
$\theta_Y-\bar L\theta_D=\Cov\!\big(h(X),\,m_Y(X)-\bar L\,m_D(X)\big)
=\E\!\big[h(X)\{Y-\bar L D\}\big]$,
\begin{equation}\label{eq:exactnull}
H_0^{\ast}:\ \E\!\big[h(X)\{Y-\bar L\,D\}\big]=0,
\end{equation}
which is necessary and sufficient for $\beta_{2SLS}=\bar L$ whenever $S_D\neq0$.
Written as a ratio, $H_0^{\ast}$ invites a plug-in of
$\bar L=S_Y/S_D$, fragile precisely where contamination is worst. Clearing the
denominator gives an equivalent restriction in objects that are all
$\sqrt n$-estimable irrespective of the signal. With
$\psi\equiv\theta_Y S_D-\theta_D S_Y$,
\begin{equation}\label{eq:psi}
b=\beta_{2SLS}-\bar L=\frac{\psi}{S_D\,\E[\Ztil D]},
\qquad H_0^{\ast}\iff\psi=0 .
\end{equation}
The two forms test the same null but differ in where the small signal $S_D$ appears,
and \eqref{eq:psi} forces a genuine trade-off. We build both, a plug-in that tracks
$b$ directly and a signal-cleared form that tracks $\psi$.

For the plug-in we estimate $H_0^{\ast}$ with the bias-corrected $\hat{\bar L}$ of
\eqref{eq:correction} (Section~\ref{sec:correction}) and the flexible curvature
$\hat h_i$. Writing
$\hat g=n^{-1}\sum_i \hat h_i\{Y_i-\hat{\bar L}D_i\}
=\hat\theta_Y-\hat{\bar L}\hat\theta_D=\hat\psi/\hat S_D$, the statistic is
$T_n^{\ast}=\hat g^2/\hat v$ with $\hat v$ the pairs-bootstrap variance of $\hat g$,
obtained by recomputing $\hat h$ and $\hat{\bar L}$ on each resample. Suppose
$H_0^{\ast}$ and Assumption~\ref{as:test} hold and $S_D$ is bounded away from zero, so
that $\hat{\bar L}$ is $\sqrt n$-consistent for $\bar L$
(Proposition~\ref{prop:correction}). Then $\sqrt n\,\hat g\convd N(0,v)$ for a finite
$v>0$, the pairs-bootstrap variance satisfies $n\hat v\convp v$, and
$T_n^{\ast}=n\hat g^2/(n\hat v)\convd\chi^2_1$, consistent against any fixed
alternative with nonzero bias. The representation $\hat g=\hat\psi/\hat S_D$ shows the
weakness. As $S_D\to0$ the plug-in $\hat{\bar L}$ is no longer $\sqrt n$-consistent and
$T_n^{\ast}$ divides by a vanishing $\hat S_D$, so it is no longer asymptotically
pivotal (Proposition~\ref{prop:tradeoff}). This is the size side of the trade-off in
\eqref{eq:psi}.

The $\psi$-form removes exactly that boundary problem by testing $\hat\psi$ itself,
which never divides by the signal. Estimate the four reduced-form means from the
flexible first-step regression of $Z$ on $(1,X',b(X)')$, with
$\hat\theta_W=n^{-1}\sum_i\hat h_i W_i$ and
$\hat S_W=n^{-1}\sum_i(Z_i-\hat e(X_i))W_i$ for $W\in\{D,Y\}$, and form
\begin{equation}\label{eq:arstat}
\hat\psi=\hat\theta_Y\hat S_D-\hat\theta_D\hat S_Y,\qquad
A_n=\hat\psi^2/\hat v_\psi,
\end{equation}
where $\hat v_\psi$ is a consistent estimator of $\Var(\hat\psi)$, either the
delta-method plug-in $n^{-1}\widehat{\nabla\psi}{}'\hat\Omega\,\widehat{\nabla\psi}$
with $\nabla\psi=(-S_Y,S_D,\theta_Y,-\theta_D)'$ and $\hat\Omega$ the joint
variance of $\sqrt n\,(\hat\theta_D,\hat\theta_Y,\hat S_D,\hat S_Y)$, or the pairs
bootstrap. Because $\hat\psi$ is a smooth function of averages that are all
$\sqrt n$-estimable without dividing by the first stage, $A_n$ stays pivotal at the
collinearity boundary.

\begin{proposition}[Weak-identification-robust bias targeting]\label{prop:arbias}
Let $\{P_n\}$ satisfy Assumption~\ref{as:test} uniformly, with the curvature
nondegenerate, $h_n\to h^\ast\neq0$ and hence $(\theta_D^\ast,\theta_Y^\ast)\neq0$,
and let $\hat v_\psi$ be ratio-consistent for $\Var(\hat\psi)$ along the sequence.
Then under $H_0^{\ast}$, $A_n\convd\chi^2_1$ whether or not $S_{D,n}$ is bounded away
from zero, so the test that rejects when $A_n>\chi^2_{1,1-\alpha}$ has asymptotic size
$\alpha$ uniformly over the collinearity range. Against a fixed alternative with
$b=\beta_{2SLS}-\bar L\neq0$ and $S_D$ bounded away from zero the power tends to one.
As $S_D\to0$ with $b$ fixed, $\psi=b\,S_D\,\E[\Ztil D]\to0$ and the power against that
fixed $b$ declines toward $\alpha$, the unavoidable weak-identification trade-off
familiar from Anderson--Rubin inference.
\end{proposition}

The two statistics are complementary by construction, not by hedge. $A_n$ delivers
size uniformly, including the collinearity boundary where $T_n^{\ast}$ divides by a
vanishing signal, and $T_n^{\ast}$ buys higher power where the design is strongly
identified. This sits alongside the earlier axis of the ``report both'' protocol, the
screening test of Section~\ref{sec:test}, whose size is likewise controlled under weak
identification (Proposition~\ref{prop:weakrobust}) and which serves as a cheap
sufficient first pass, against the bias-targeting test that adjudicates harmless
curvature. The protocol is to report the screen and the bias-targeting test alongside
the contamination share $\hat\varrho$, which by Corollary~\ref{cor:frontier} says how
much identification a correction would cost, using $A_n$ for a size guarantee that
holds across the collinearity range and $T_n^{\ast}$ where strong identification makes
its extra power worth having. In the finite-sample designs of Section~\ref{sec:head}
the two forms agree when strongly identified and $A_n$ holds size as
$\Var(Z\mid X)\to0$ (Table~\ref{tab:svx}), confirming
Proposition~\ref{prop:arbias}. Section~\ref{sec:head} compares both against a RESET.

\subsection{The over-identified case}\label{sec:overid}

The bias-targeting restriction has an exact over-identified analogue, at the cost of the
single reduced form. Let $Z=(Z_1,\dots,Z_J)'$ be a vector of instruments with residualized
components $\tilde Z_j=Z_j-L(Z_j\mid1,X)$, and collect the second moments
$Q_{ZD}=\E[\tilde Z D]$, $Q_{ZY}=\E[\tilde Z Y]$, and $Q_{ZZ}=\E[\tilde Z\tilde Z']$. Each
inherits the decomposition of Lemma~\ref{lem:decomp}, $Q_{ZD}=S_D+\theta_D$ and
$Q_{ZY}=S_Y+\theta_Y$ with per-instrument signal and contamination vectors
$S_{W,j}=\E[\Cov(Z_j,W\mid X)]$ and $\theta_{W,j}=\E[h_j(X)W]$, while
$Q_{ZZ}=\Sigma_S+H$ with $\Sigma_S=\E[\Cov(Z\mid X)]$ and $H=\E[h(X)h(X)']$ the curvature
Gram matrix. The population 2SLS coefficient on $D$ is the ratio of quadratic forms
\begin{equation}\label{eq:oid2sls}
\beta_{2SLS}=\frac{Q_{ZD}'\,Q_{ZZ}^{-1}\,Q_{ZY}}{Q_{ZD}'\,Q_{ZZ}^{-1}\,Q_{ZD}},
\end{equation}
which is just-identified IV with the combined instrument $\hat D=\pi'\tilde Z$,
$\pi=Q_{ZZ}^{-1}Q_{ZD}$. Let $\bar L$ be the benchmark of Section~\ref{sec:correction}, the
2SLS estimand built from the clean instruments $Z_j-e_j(X)$, which specializes
\eqref{eq:oid2sls} with $(S_D,S_Y,\Sigma_S)$ in place of $(Q_{ZD},Q_{ZY},Q_{ZZ})$ and
reduces to $S_Y/S_D$ when $J=1$.

\begin{proposition}[Over-identified screen and exact restriction]\label{prop:overid}
(i) If $\theta_D=\theta_Y=0$ then $\E[\hat D\,W]=\pi'S_W$ for $W\in\{D,Y\}$, so the combined
instrument's curvature part $\pi'h(X)$ is orthogonal to $D$ and $Y$ and $\beta_{2SLS}$ is a
contamination-free instrument combination whatever the curvature $H$ carried by the
weighting matrix. (ii) With $\E[\hat D D]\neq0$,
\begin{equation}\label{eq:oidexact}
\beta_{2SLS}=\bar L\quad\Longleftrightarrow\quad
H_0^{\dagger}:\ \E\big[\hat D\,(Y-\bar L D)\big]=0 .
\end{equation}
(iii) At $J=1$, $\hat D\propto\tilde Z$ and $S_Y-\bar L S_D=0$, so \eqref{eq:oidexact}
becomes $\E[h(X)\{Y-\bar L D\}]=0$ of \eqref{eq:exactnull}.
\end{proposition}

\begin{proof}
For (i), $\E[\hat D W]=\pi'Q_{ZW}=\pi'(S_W+\theta_W)=\pi'S_W$ when $\theta_W=0$, and since
$\hat D=\pi'(Z-e(X))+\pi'h(X)$ the second term contributes $\pi'\theta_W=0$ to each
moment. For (ii), $\beta_{2SLS}=\E[\hat D Y]/\E[\hat D D]$ by \eqref{eq:oid2sls}, so
$\beta_{2SLS}-\bar L=\E[\hat D(Y-\bar L D)]/\E[\hat D D]$. For (iii), at $J=1$ the scalar
$\pi$ cancels and $\bar L=S_Y/S_D$ gives $S_Y-\bar L S_D=0$, leaving
$\theta_Y-\bar L\theta_D=\E[h(X)\{Y-\bar L D\}]$.
\end{proof}

The screen $\theta_D=\theta_Y=0$ therefore stays a \emph{sufficient} test for the absence
of contamination, now for a transparent reason. Any combination of instruments whose
per-component curvature is orthogonal to $D$ and $Y$ is itself a clean instrument, so the
weighting $\pi$ may be contaminated by $H$ without biasing the estimand, and the
finite-sample behaviour of the stacked $\chi^2_{2J}$ screen is the design~(C) evidence of
Section~\ref{sec:overidMC}. What the screen cannot do, exactly as in the just-identified
case, is separate harmful from harmless curvature, and \eqref{eq:oidexact} is the
restriction that does. Written out it is $\pi'(Q_{ZY}-\bar L Q_{ZD})=0$, and unlike the
just-identified null $\E[h(X)\{Y-\bar L D\}]=0$ it does not reduce to a statement about the
curvature alone. The combination $\pi=Q_{ZZ}^{-1}Q_{ZD}$ carries the contamination through
both the moment vector $Q_{ZD}$ and the inverse second-moment matrix $Q_{ZZ}^{-1}$, so the
exact analogue must invert $Q_{ZZ}$ and couples the signal-side over-identification
weighting with the contamination in a way the single reduced form does not. This is the
sense in which over-identification is the less mechanical direction.

The sample analogue plugs the residualized second moments and the corrected
$\hat{\bar L}$ of \eqref{eq:correction} into \eqref{eq:oidexact}, forming
$\hat g=n^{-1}\sum_i\hat D_i(Y_i-\hat{\bar L}D_i)$ and $T_n^{\dagger}=\hat g^2/\hat v$ with
$\hat v$ the pairs-bootstrap variance. As a smooth function of sample second moments
$\hat g$ is asymptotically normal when the combined first stage $\E[\hat D D]$ is bounded
away from zero, so $T_n^{\dagger}\convd\chi^2_1$ under $H_0^{\dagger}$, and clearing
$\E[\hat D D]$ from the denominator gives a signal-cleared form valid at the collinearity
boundary exactly as $A_n$ does for the just-identified $\psi$ in \eqref{eq:psi}.

\subsection{Reporting the contamination share}\label{sec:share}
Beyond a hypothesis test, the decomposition yields a single interpretable number.
Define the \emph{contamination share}
\begin{equation}\label{eq:share}
\varrho\;\equiv\;\frac{\theta_D}{S_D+\theta_D}\;=\;\frac{\theta_D}{\E[\Ztil D]},
\end{equation}
the fraction of the observed first-stage \emph{covariance} $\E[\Ztil D]$
attributable to propensity curvature rather than to the within-cell signal. It is a
share of the covariance, not of the conditional $F$ or the concentration parameter,
which are quadratic in it. In the nonlinear-propensity family of
Section~\ref{sec:mc}, \eqref{eq:share} reduces to
$\varrho=\E[h^2]/\Var(\Ztil)\in[0,1]$. In general $\varrho$ is a \emph{signed}
decomposition of the first-stage covariance. Since $S_D>0$ always but $\theta_D$ is a
signed covariance, $\varrho\in[0,1]$ and curvature \emph{inflates} the $F$ only under
$\theta_D>0$, the sign the mode-(a) designs here satisfy. When $\theta_D<0$ the
curvature instead deflates the $F$ and gives $\varrho<0$, and $\varrho$ is unstable when
signal and curvature nearly cancel ($\E[\Ztil D]\approx0$). We therefore report $\varrho$
with
its sign. It is estimated by
$\hat\varrho=\hat\theta_D/\widehat{\E[\Ztil D]}$ with
$\widehat{\E[\Ztil D]}=n^{-1}\sum_i \Ztil_i D_i$, and a percentile bootstrap
furnishes a confidence interval. We recommend reporting $\hat\varrho$ next to the
conditional $F$. A large $F$ paired with a non-negligible $\hat\varrho$ signals
that the strength is partly manufactured, the signature of mode~(a). Yet $\varrho$ is a first-stage gauge and measures only mode~(a). It is \emph{not} a
bias estimate, because by Proposition~\ref{prop:impossible} the bias also loads on
the outcome-side term $\theta_Y$. A small $\hat\varrho$ therefore does not certify a
design. When it is paired with a directed-test rejection, that combination is
itself the signature of mode~(b), namely honest strength with outcome-side bias, as
in the health insurance design, where $\hat\varrho=0.09$ yet the estimand more than
doubles once corrected.
The share thus complements, and does not replace, the outcome-based tests of
Sections~\ref{sec:test}--\ref{sec:exact}. Read together, a large $F$, the share
$\hat\varrho$, and the directed-test $p$-value locate which mode, if either, is in
play.

%=============================================================================
\section{Bias-corrected estimation and the identification trade-off}\label{sec:correction}
%=============================================================================

The decomposition suggests an immediate correction, namely to subtract the
estimated contamination from numerator and denominator. Using
$\E[\Ztil Y]-\theta_Y=\E[(Z-e(X))Y]$ and likewise for $D$, the corrected estimand
equals the signal ratio $\bar L$. It is implemented by replacing the
linear-projection residual $\Ztil$ with the residual from a \emph{flexibly}
estimated propensity $\hat e(\cdot)$,
\begin{equation}\label{eq:correction}
\hat{\bar L}=\frac{\sum_{i}\big(Z_i-\hat e(X_i)\big)Y_i}
                  {\sum_{i}\big(Z_i-\hat e(X_i)\big)D_i}.
\end{equation}

\begin{proposition}[Consistency of the corrected estimator]\label{prop:correction}
Suppose $\hat e(\cdot)$ satisfies $\|\hat e-e\|_{L^2}=o_p(1)$ and, in the fixed
data-generating process, $S_D=\E[\Cov(Z,D\mid X)]\neq0$. Then under
Assumptions~\ref{as:lin}--\ref{as:reg},
$\hat{\bar L}\convp \bar L=S_Y/S_D$, the saturated specification estimand, a convex
average of within-cell IV slopes and, in the binary specialization of
Remark~\ref{rem:binary}, a convex average of conditional LATEs. Consistency
requires no orthogonality or rate condition. Asymptotic normality and inference do (see the
Appendix and Proposition~\ref{prop:orthogonal}), and, since $S_D$ may be small,
weak-instrument-robust confidence sets are appropriate.
\end{proposition}

Estimator~\eqref{eq:correction} is the residualized-propensity, or
partialling-out, IV estimator. It coincides with the recommendation of
\citet{blandhol2022} to saturate the covariates, here written in
influence-function form. Our contribution is not the estimator but the following
statement of its cost, which the collinearity limit makes precise.

\begin{proposition}[No free lunch]\label{prop:tradeoff}
The corrected estimator's own first-stage covariance is exactly the signal term,
$\E[(Z-e(X))D]=S_D=\E[\Cov(Z,D\mid X)]$. Under the collinearity sequence of
Assumption~\ref{as:seq}, $S_D\to0$. Hence the corrected estimator is weakly
identified in precisely the regime in which the correction matters most. There
is no specification of the covariates that both removes the contamination and
preserves the apparent first-stage strength, because the strength was the contamination.
\end{proposition}

\begin{corollary}[The identification surrendered equals the contamination removed]\label{cor:frontier}
The corrected and reported (contaminated) first-stage covariances obey the exact
identity
\[
\underbrace{S_D}_{\text{corrected first stage}}
=(1-\varrho)\,\underbrace{\E[\Ztil D]}_{\text{reported first stage}},
\qquad \varrho=\frac{\theta_D}{\E[\Ztil D]},
\]
where $\varrho$ is the contamination share \eqref{eq:share}. De-contamination
scales the identifying covariance down by exactly $\varrho$, and the corrected
concentration parameter $S_D^2/\{\sigma_v^2\,\E[\Var(Z\mid X)]\}$ carries a factor
$(1-\varrho)^2$ in its numerator. The trade-off of
Proposition~\ref{prop:tradeoff} is therefore quantitative and known in advance.
The fraction of identifying variation surrendered by correcting is the same
$\varrho$ that measured how much of the apparent strength was contamination. A
design with small $\varrho$, such as the health insurance instrument of
Section~\ref{sec:hi} ($\hat\varrho=0.09$), is corrected at negligible cost,
whereas a design whose strength is mostly curvature ($\varrho\to1$) cannot be both
corrected and strongly identified.
\end{corollary}

The practical message follows. The analyst
facing a strong conditional $F$ under collinearity between instrument and covariates has not
found a strong design that merely needs a better second-stage specification. The
strength is an artifact of the very nonlinearity that biases the estimand, and
removing it returns the design to (correctly) weak. The response is to diagnose
with our test, followed by either a redesigned instrument or explicit
weak-IV-robust inference on the corrected estimand
\citep{andrews2019,montielolea2013}.

%=============================================================================
\section{Monte Carlo evidence}\label{sec:mc}
%=============================================================================

\subsection{Designs}
Throughout, $X\sim\text{Uniform}(-2,2)$ and $g(X)=X^2-\tfrac43$ (mean zero), the
per-unit structural effect is $\tau(X)=1+0.25X\in[0.5,1.5]$, so the hull of within-cell
IV slopes is $[0.5,1.5]$, and the outcome is $Y=A_\mu\,g(X)+\tau(X)D+\mathcal N(0,1)$
with $A_\mu=1.5$. We use four designs that span the support the estimator meets in
practice. \textbf{(A)} An \emph{ordered} treatment
$D\in\{0,1,2,3\}$, built from three independent monotone stages
$D=\sum_{j=1}^3\mathbf 1\{V_j\le\pi_{0j}+Z\,\delta_j\}$ with a binary instrument
$Z\sim\text{Bernoulli}(e(X))$ and $e(X)=\Lambda(\kappa\,g(X))$, where the knob $\kappa$
drives $\Var(Z\mid X)\to0$. \textbf{(B)} A \emph{continuous} instrument
$Z=e(X)+\sigma\,\mathcal N(0,1)$ with $e(X)=0.3X+0.6\,g(X)$ and a \emph{continuous}
treatment $D=0.5+0.5\,g(X)+(Z-e(X))+0.3\,\mathcal N(0,1)$, so $\Var(Z\mid X)=\sigma^2$ is
the collinearity knob and $\Cov(Z,D\mid X)=\sigma^2$ the signal. \textbf{(C)} An
\emph{over-identified} design with two binary instruments whose propensities
$e_1,e_2$ are nonlinear in $X$ and a binary treatment monotone in both. \textbf{(D)} A
binary instrument and binary treatment, $e(X)=\Lambda(\kappa g(X))$ with a binary
complier structure, the case of Remark~\ref{rem:binary}, used for the finite-sample
size, power, and head-to-head exercises below. The within-cell covariance
$\Cov(Z,D\mid X)$ is of a single sign by construction in every design, so the sign
condition of Section~\ref{sec:framework} holds and the within-cell hull $[0.5,1.5]$ is
the correct convex benchmark. All numbers are from the \texttt{R} scripts released with
the paper.

\subsection{The collinearity limit}
Table~\ref{tab:collin} traces designs (A) and (B) as $\Var(Z\mid X)\to0$, with
$n=\num{200000}$ to approximate the population. In both, the 2SLS
estimand leaves the within-cell slope hull $[0.5,1.5]$ already at mild collinearity and
converges monotonically to the theoretical contamination limit
$\Cov(h,m_Y)/\Cov(h,m_D)$ of Proposition~\ref{prop:collin} (which the ordered design
approaches from below as its own limit falls, and the continuous design hits exactly at
$4.0$). Throughout, the conditional first-stage $F$ stays enormous, on the order of
$\num{12000}$ in (A) and in the millions in (B), never warning, confirming
Corollary~\ref{cor:Ffails}. Figure~\ref{fig:collin} plots both. On a single contaminated
dataset from each design ($n=\num{100000}$) the directed test rejects decisively, with
screening and bias-targeting $p$-values below $10^{-4}$, and the bias-corrected estimator
recovers the within-cell benchmark, $0.99\,[0.94,1.05]$ in (A) and $1.00\,[0.96,1.04]$ in
(B).

\begin{table}[t]\centering
\caption{The collinearity limit ($n=\num{200000}$). In panel~(A)
the treatment is ordered, $D\in\{0,1,2,3\}$, and the instrument binary. In panel~(B) both
the instrument and the treatment are continuous. In each the 2SLS estimand leaves the
within-cell slope hull $[0.5,1.5]$ and converges to the contamination limit
$\Cov(h,m_Y)/\Cov(h,m_D)$, while the conditional $F$ stays orders of magnitude above the
threshold of $10$.}
\label{tab:collin}
\begin{tabular}{rrrrr}
\toprule
\multicolumn{5}{l}{\emph{Panel (A). Ordered treatment $D\in\{0,1,2,3\}$, binary instrument}}\\
$\kappa$ & $\beta_{2SLS}$ & cond.\ $F$ & contam.\ limit & $\E[\Var(Z\mid X)]$\\
\midrule
0.10 & 1.530 & \num{12570} & 151.4 & 0.249\\
0.50 & 3.539 & \num{12635} & 32.8  & 0.230\\
1.00 & 5.377 & \num{12821} & 19.1  & 0.189\\
2.00 & 7.241 & \num{12890} & 13.1  & 0.119\\
4.00 & 8.415 & \num{12666} & 10.7  & 0.057\\
8.00 & 8.844 & \num{12187} & 9.78  & 0.028\\
32.0 & 8.857 & \num{12473} & 9.10  & 0.007\\
\midrule
\multicolumn{5}{l}{\emph{Panel (B). Continuous instrument and continuous treatment}}\\
$\sigma$ & $\beta_{2SLS}$ & cond.\ $F$ & contam.\ limit & $\E[\Var(Z\mid X)]$\\
\midrule
1.00 & 1.897 & $2.8{\times}10^{6}$ & 4.00 & 1.000\\
0.60 & 2.623 & $1.5{\times}10^{6}$ & 4.00 & 0.360\\
0.40 & 3.176 & $1.1{\times}10^{6}$ & 4.00 & 0.160\\
0.25 & 3.616 & $9.1{\times}10^{5}$ & 4.00 & 0.062\\
0.15 & 3.859 & $8.2{\times}10^{5}$ & 4.00 & 0.022\\
0.05 & 3.974 & $7.9{\times}10^{5}$ & 4.00 & 0.003\\
0.02 & 3.993 & $8.0{\times}10^{5}$ & 4.00 & 0.000\\
\bottomrule
\end{tabular}
\end{table}

\begin{figure}[t]\centering
\includegraphics[width=\textwidth]{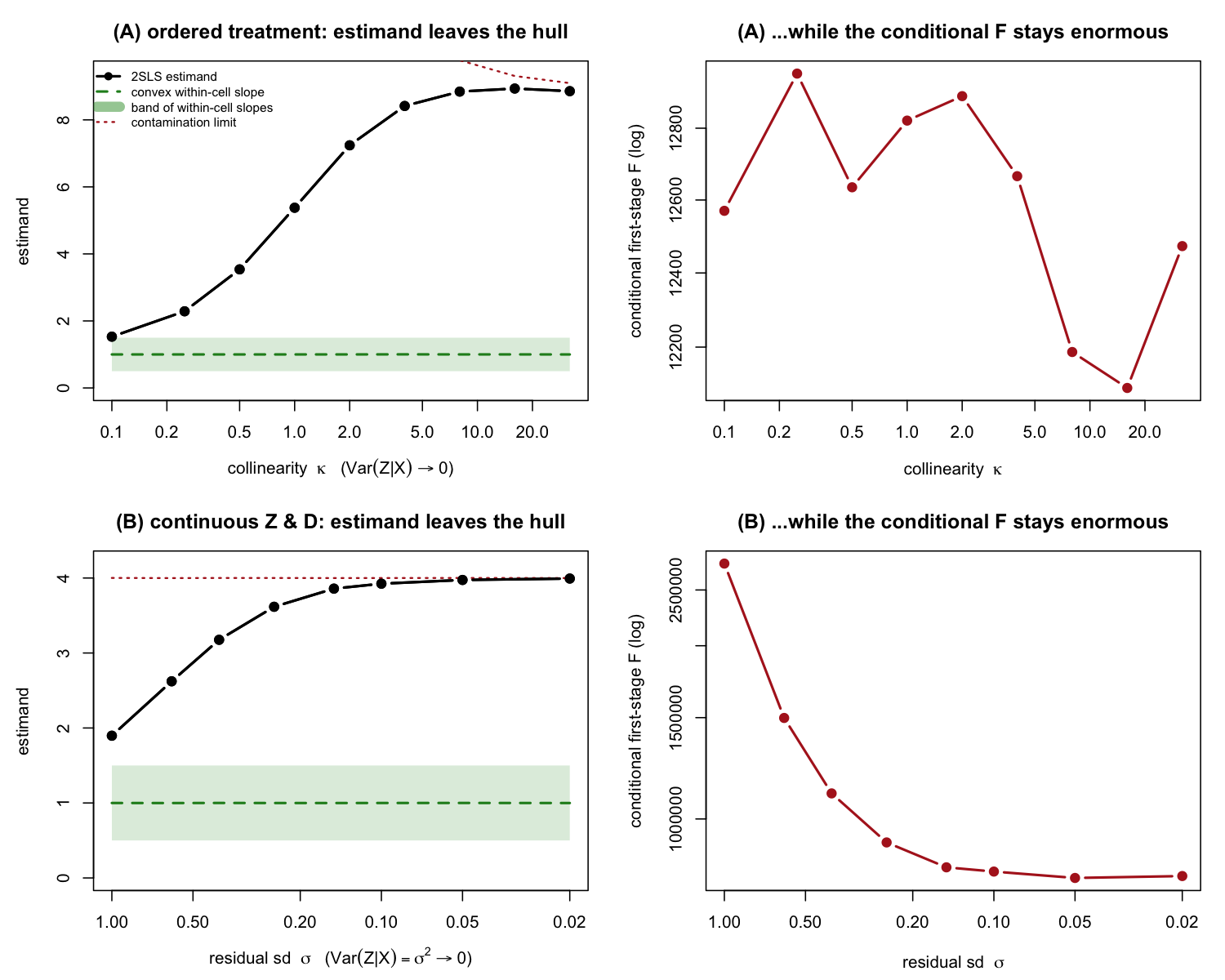}
\caption{Top row, ordered treatment (A). Bottom row, continuous instrument and treatment
(B). In the left panels, as $\Var(Z\mid X)\to0$ the 2SLS estimand (solid) leaves the band
of within-cell IV slopes (shaded, $[0.5,1.5]$) and departs from the convex within-cell
slope (dashed), converging to the contamination limit (dotted). In the right panels the
conditional first-stage $F$ (log scale) stays enormous throughout, far above the weak-IV
threshold of $10$.}
\label{fig:collin}
\end{figure}

\subsection{The binary case, where the estimand leaves the LATE hull}\label{sec:mcbinary}
Design (D) makes the collinearity limit concrete in the model where the benchmark has a
causal reading. Here the within-cell slopes are the conditional LATEs, so the hull
$[0.5,1.5]$ is the LATE hull, and $\bar L$ is a convex average of conditional
LATEs. Table~\ref{tab:collinbinary} traces the binary design as $\kappa$ rises, with
$n=\num{200000}$. The 2SLS estimand leaves the LATE hull already at mild collinearity and
climbs to $8.86$, far above $\sup_x\late(x)=1.5$, even though every conditional LATE is
positive and bounded by $1.5$. It converges to the contamination limit
$\Cov(h,m_Y)/\Cov(h,m_D)$, which is itself outside the hull, while the conditional first-stage
$F$ stays near $\num{40000}$ throughout, four orders of magnitude above the threshold of
$10$. This is the construction behind Proposition~\ref{prop:collin}, an instance in which
all $\late(x)>0$ yet the limit exceeds $\sup_x\late(x)$. Figure~\ref{fig:collinbinary} plots
the estimand leaving the band of all conditional LATEs against the unmoved $F$.

\begin{table}[t]\centering
\caption{The binary case ($n=\num{200000}$, $A_\mu=1.5$), design (D). True
$\late(X)\in[0.5,1.5]$ throughout. The 2SLS estimand leaves the LATE hull and converges to
the contamination limit $\Cov(h,m_Y)/\Cov(h,m_D)$, while the conditional $F$ stays near
$\num{40000}$.}
\label{tab:collinbinary}
\begin{tabular}{rrrrr}
\toprule
$\kappa$ & $\beta_{2SLS}$ & cond.\ $F$ & contam.\ limit & $\E[e(1-e)]$\\
\midrule
0.10 & 1.527 & \num{39808} & 151.4 & 0.249\\
0.25 & 2.316 & \num{39888} & 61.9  & 0.245\\
0.50 & 3.497 & \num{40459} & 32.8  & 0.230\\
1.00 & 5.408 & \num{38776} & 19.1  & 0.189\\
2.00 & 7.325 & \num{38904} & 13.1  & 0.118\\
4.00 & 8.487 & \num{37790} & 10.7  & 0.057\\
8.00 & 8.794 & \num{37881} & 9.77  & 0.027\\
16.0 & 8.877 & \num{37979} & 9.32  & 0.013\\
32.0 & 8.860 & \num{38215} & 9.11  & 0.007\\
\bottomrule
\end{tabular}
\end{table}

\begin{figure}[t]\centering
\includegraphics[width=\textwidth]{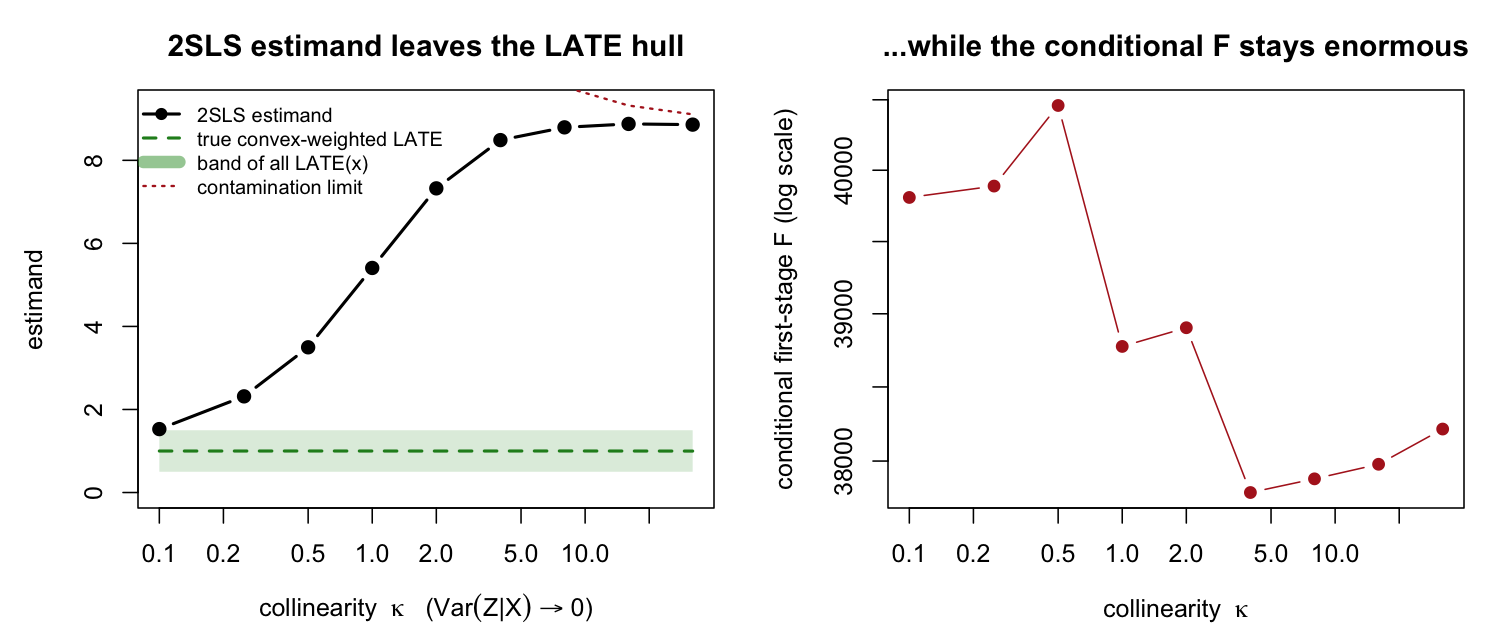}
\caption{The binary design (D). Left panel. As collinearity $\kappa$ rises, the 2SLS
estimand (solid) leaves the band of all conditional LATEs (shaded) and departs from the
true convex average of LATEs (dashed), converging to the contamination limit (dotted). Right
panel. The conditional first-stage $F$ (log scale) stays near $\num{40000}$ throughout, far
above the weak-IV threshold of $10$.}
\label{fig:collinbinary}
\end{figure}

\subsection{Over-identification and the stacked screen}\label{sec:overidMC}
Design (C) tests the over-identified extension. With two instruments the screen stacks the
four contamination moments $(\theta_{D,1},\theta_{Y,1},\theta_{D,2},\theta_{Y,2})$ and
refers the Wald statistic to $\chi^2_4$. Over $\num{2000}$ replications at $n=3000$, under
linear propensities (the null) it rejects $1.5\%$ of the time at the $5\%$ level, and under
curved propensities its power is $1.00$, while the joint first-stage $F$ of the two
instruments averages $131$ under the null and $246$ under the alternative. As in the
just-identified case the screen separates designs the first-stage $F$ cannot.

The same design tests the exact restriction \eqref{eq:oidexact}. Alongside the contaminated
design and a linear-propensity clean design, we add a \emph{harmless} over-identified
design, curved propensities with a linear outcome level and a homogeneous effect, where
$\theta_D,\theta_Y\neq0$ yet $\beta_{2SLS}=\bar L$. Table~\ref{tab:overid} reports rejection
rates. The stacked screen fires on both the contaminated and the harmless designs, because
curvature is present in each, while the exact test $T_n^{\dagger}$ of \eqref{eq:oidexact}
has power one against the contaminated design ($\beta_{2SLS}=9.5$ against $\bar L=1.0$) and
holds size on the harmless ($4.3\%$) and clean ($5.7\%$) designs, where
$\beta_{2SLS}=\bar L$. It is the over-identified counterpart of the head-to-head in
Table~\ref{tab:head}, separating harmful from harmless curvature exactly where the stacked
screen and the joint first-stage $F$ cannot.

\begin{table}[t]\centering
\caption{Over-identified head-to-head (two instruments, $\num{2000}$ replications, $n=3000$,
cubic basis, $B=199$). The stacked screen tests $\theta_D=\theta_Y=0$ ($\chi^2_4$); the
exact test $T_n^{\dagger}$ tests $\beta_{2SLS}=\bar L$ via \eqref{eq:oidexact} ($\chi^2_1$).
Rejection rates at the $5\%$ level. The harmless design has curvature present
($\theta\neq0$) but no bias ($\beta_{2SLS}=\bar L$), the case only the exact test resolves.}
\label{tab:overid}
\begin{tabular}{lcccc}
\toprule
design & stacked screen & exact $T_n^{\dagger}$ & $\beta_{2SLS}$ & $\bar L$\\
\midrule
CONTAMINATED (bias $\neq0$)     & 1.00  & 1.00           & 9.51 & 1.00\\
HARMLESS \ \ \ \ (bias $=0$)    & 1.00  & \textbf{0.043} & 1.00 & 0.99\\
CLEAN \ \ \ \ \ \ \ \ (bias $=0$) & 0.016 & 0.057        & 1.01 & 1.00\\
\bottomrule
\end{tabular}
\end{table}

\subsection{Size and power}
In design (D), the binary instrument and treatment, two exercises confirm the test's
finite-sample behaviour. Both are collected in \ref{app:mc}. On a single large contaminated dataset ($n=\num{100000}$,
$\kappa=4$) the directed test rejects at $p<10^{-4}$ while the conditional $F$ is
$\num{19737}$, and the bias-corrected estimator recovers the true convex average of LATEs (Table~\ref{tab:sizepower}). Over $\num{2000}$ replications at $n=3000$ it has a
conservative size of $1.9\%$ against the $5\%$ level under a linear propensity, and
power essentially one under a curved one, at
a mean conditional $F$ near $585$ in \emph{both} cases. The $F$ cannot discriminate
the two designs the test separates almost perfectly. A power sweep
(Table~\ref{tab:power}, Figure~\ref{fig:power}) raises the induced bias from $0$ to
$2.64$ while the mean conditional $F$ stays pinned near $\num{1740}$.

\subsection{Head-to-head comparison across three designs}\label{sec:head}
Still in design (D), Table~\ref{tab:head} and Figure~\ref{fig:head} compare three
procedures across three designs at $n=3000$. The first procedure is a Ramsey RESET testing
curvature of the propensity $Z\sim X$, a functional of $(Z,X)$ alone. The second
is the strongly identified screening test of Section~\ref{sec:test}
($H_0:\theta_D=\theta_Y=0$), which we label ``Directed-S''. The third is the
bias-targeting test \eqref{eq:exactnull} of Section~\ref{sec:exact}, labelled
``Directed-X''. The \emph{contaminated} design has a nonlinear propensity and a
nonlinear outcome, so it carries real bias. The \emph{harmless} design has a
nonlinear propensity but a linear outcome and a homogeneous effect, so the bias
is zero despite strong curvature. The \emph{clean} design has a linear propensity
and zero bias. All three designs carry a mean conditional $F$ near $600$.

The pattern illustrates Proposition~\ref{prop:impossible}. In the
harmless design the RESET rejects $100\%$ of the time. It detects the propensity
curvature but, being blind to the outcome, cannot tell that the curvature is
inconsequential, and the same $F\approx600$ that accompanies a large bias in the
contaminated design accompanies zero bias here. Directed-S also raises false alarms
($100\%$), because it tests the sufficient, not the necessary, condition. Only
Directed-X separates the three designs as the theory requires. It has full power
against genuine contamination, holds size in the harmless design ($4.7\%$ in the plug-in
form $T_n^{\ast}$ and $5.0\%$ in the signal-cleared form $A_n$) that the other two
procedures and the conditional $F$ cannot certify, and sits at $4.1\%$ in the clean
design, within sampling error of the nominal level. The two forms coincide here because
all three designs are strongly identified, as Proposition~\ref{prop:arbias} anticipates.
The clean-design size holds across the sample sizes of Table~\ref{tab:nsweep}, from
$5.8\%$ at $n=1000$ to $4.4\%$ at $n=\num{20000}$, and is correctly calibrated rather than
merely asymptotic, with the sensitivity to the number of studentizing bootstrap draws
detailed in that table.

\begin{table}[t]\centering
\caption{Head-to-head rejection rates ($5\%$ level, $\num{3000}$ replications,
$n=3000$, cubic basis, $B=999$). Directed-X is the plug-in $T_n^{\ast}$ and
Directed-X (AR) the signal-cleared $A_n$ of \eqref{eq:arstat}. ``mean cond.\ $F$'' is
the average Sanderson--Windmeijer conditional first-stage $F$. Monte Carlo standard
errors are below $0.004$.}
\label{tab:head}
\begin{tabular}{lccccc}
\toprule
design & RESET & Directed-S & Directed-X & Directed-X (AR) & mean cond.\ $F$\\
\midrule
CONTAMINATED (bias $\neq0$) & 1.00 & 1.00 & 1.00 & 1.00 & 605\\
HARMLESS \ \ \ (bias $=0$)  & 1.00 & 1.00 & \textbf{0.047} & \textbf{0.050} & 607\\
CLEAN \ \ \ \ \ \ \ (bias $=0$) & 0.045 & 0.014 & 0.041 & 0.041 & 573\\
\bottomrule
\end{tabular}
\end{table}

\begin{figure}[t]\centering
\includegraphics[width=\textwidth]{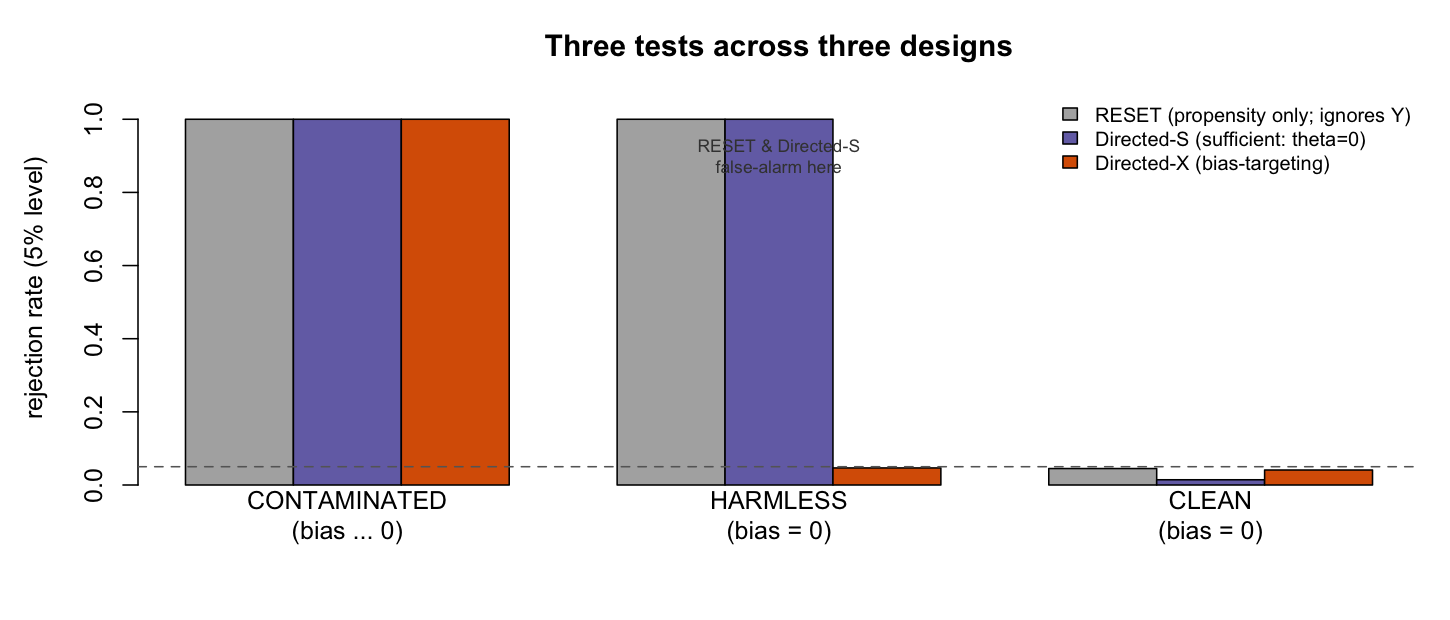}
\caption{Rejection rates of the RESET, the sufficient (Directed-S) and the
bias-targeting (Directed-X) tests across the contaminated, harmless, and clean
designs. In the harmless design (middle group) the RESET and Directed-S
raise false alarms while Directed-X holds size, and the conditional $F$ (near $600$ in
all groups) cannot discriminate.}
\label{fig:head}
\end{figure}

Table~\ref{tab:svx} in \ref{app:mc} pushes the same contrast along the
collinearity knob. Both bias-targeting forms keep full power against genuine bias as
$\Var(Z\mid X)$ falls by an order of magnitude, so neither loses power to
collinearity. They part company from the screen only in the harmless design. There
Directed-S raises a false alarm every time, because its \emph{sufficient} condition $\theta=0$
is false whenever the propensity is nonlinear, while both the plug-in Directed-X and the
signal-cleared $A_n$ control size throughout, the size guarantee for $A_n$ holding by
Proposition~\ref{prop:arbias} even at the smallest $\Var(Z\mid X)$ in the sweep. This is
the content of the ``report both'' recommendation. Directed-S is a cheap screen that, at the cubic basis used
here, misses no contamination in the probed directions but cannot vindicate harmless
curvature, and Directed-X is the adjudicator that, in those same directions,
separates harmful from harmless. Both guarantees are relative to the fixed basis, per
Remark~\ref{rem:power}. A growing or orthogonalized basis extends them to all
directions.

\subsection{Estimating the contamination share}\label{sec:shareMC}
Table~\ref{tab:share} and Figure~\ref{fig:share} in \ref{app:mc} validate
the share estimator \eqref{eq:share}. The plug-in $\hat\varrho$ at $n=3000$ tracks
the population share from $\varrho\approx0.02$ to $\varrho\approx0.75$ with
$93$--$95\%$ bootstrap coverage. Because the bias also loads on the outcome-side
term, a modest share can coincide with substantial bias. At $\varrho=0.02$ the
estimand is already displaced by $1.3$ here. Thus $\hat\varrho$ flags manufactured
strength, while the test of Section~\ref{sec:exact} adjudicates its consequences.

%=============================================================================
\section{Empirical application}\label{sec:empirics}
%=============================================================================

We work one application in full and corroborate it with two further designs. The
flagship is the effect of spousal health insurance coverage on wives' labour supply
(Section~\ref{sec:hi}), a binary-instrument, binary-treatment design in which an
overwhelming first stage and a small contamination share are both reassuring, yet a
standard linear control for husband income leaves the estimate biased by more than a
factor of two, a contamination the directed test flags and the correction repairs
without loss of identification. The 401(k) eligibility design (Section~\ref{sec:401k})
is a second binary instance of the same mode~(b). The slave trade design of
\citet{nunn2011} (Section~\ref{sec:nw}), with a continuous instrument and a continuous
treatment, is the real data face of mode~(a), manufactured strength, where correcting
the propensity collapses a first stage that looked powerful. Throughout, the flexible
propensity uses a low-order polynomial (or spline) basis and inference uses the pairs
bootstrap of Section~\ref{sec:test}.

One interpretive point governs how these results should be read, made in the
introduction and specialized here to the applications. The curvature $h=e-L(e\mid1,X)$, the
moments $\theta_D=\E[hD]$ and $\theta_Y=\E[hY]$, the share $\varrho$, and the two
directed tests are functionals of $\E[Z\mid X]$ and the level functions alone, so by
Sections~\ref{sec:test}--\ref{sec:correction} they retain their size and consistency
across these applications whatever the support of $Z$ and $D$. The binary-treatment LATE structure
of Remark~\ref{rem:binary} additionally reads the signal ratio $\bar L=S_Y/S_D$ as a
convex average of conditional LATEs through the Wald--LATE identity, which holds for the
health insurance and 401(k) designs. For the continuous slave trade design $\bar L$
is the saturated-propensity linear-IV estimand, and the test there detects
functional-form contamination of the reported estimand relative to that benchmark. The
convex decomposition for multivalued treatments is the open frontier of
Section~\ref{sec:conclusion}.

\subsection{Spousal health insurance and wives' labour supply}\label{sec:hi}

Our headline application is the effect of spousal health insurance coverage on the
weekly hours worked by married women, the design of \citet{olson1998}. The treatment
$D$ is whether a wife is covered by her husband's health insurance, the outcome $Y$ is
her weekly hours, and the instrument $Z$ is whether the husband holds insurance through
his own job, a precondition for that coverage that shifts it without bearing on the
wife's hours except through it. We use the cross-section of $n=\num{22272}$ married
women assembled by \citet{olson1998} and distributed with the \texttt{Ecdat} package.
The instrument is overwhelming, with a conditional first-stage $F$ above \num{36000}.
But husbands hold own-job coverage in a manner strongly nonlinear in their earnings,
and a husband's earnings are a first-order, nonlinear determinant of his wife's hours,
while husband income enters the standard specification linearly. This is precisely the
configuration of Lemma~\ref{lem:decomp} in which linear covariate adjustment
contaminates the estimand.

Figure~\ref{fig:hi} makes the mechanism visible. The probability that the husband
holds own-job insurance rises steeply and then saturates in his income (panel~A), so
the propensity curvature $h$ that a linear control leaves behind is large, and mean
wife hours are single-peaked in husband income (panel~B), so the outcome level $m_Y$
that curvature multiplies is itself far from linear. Their covariance is the
outcome-side contamination $\theta_Y$, and no first-stage number sees it.

\begin{figure}[t]\centering
\includegraphics[width=\textwidth]{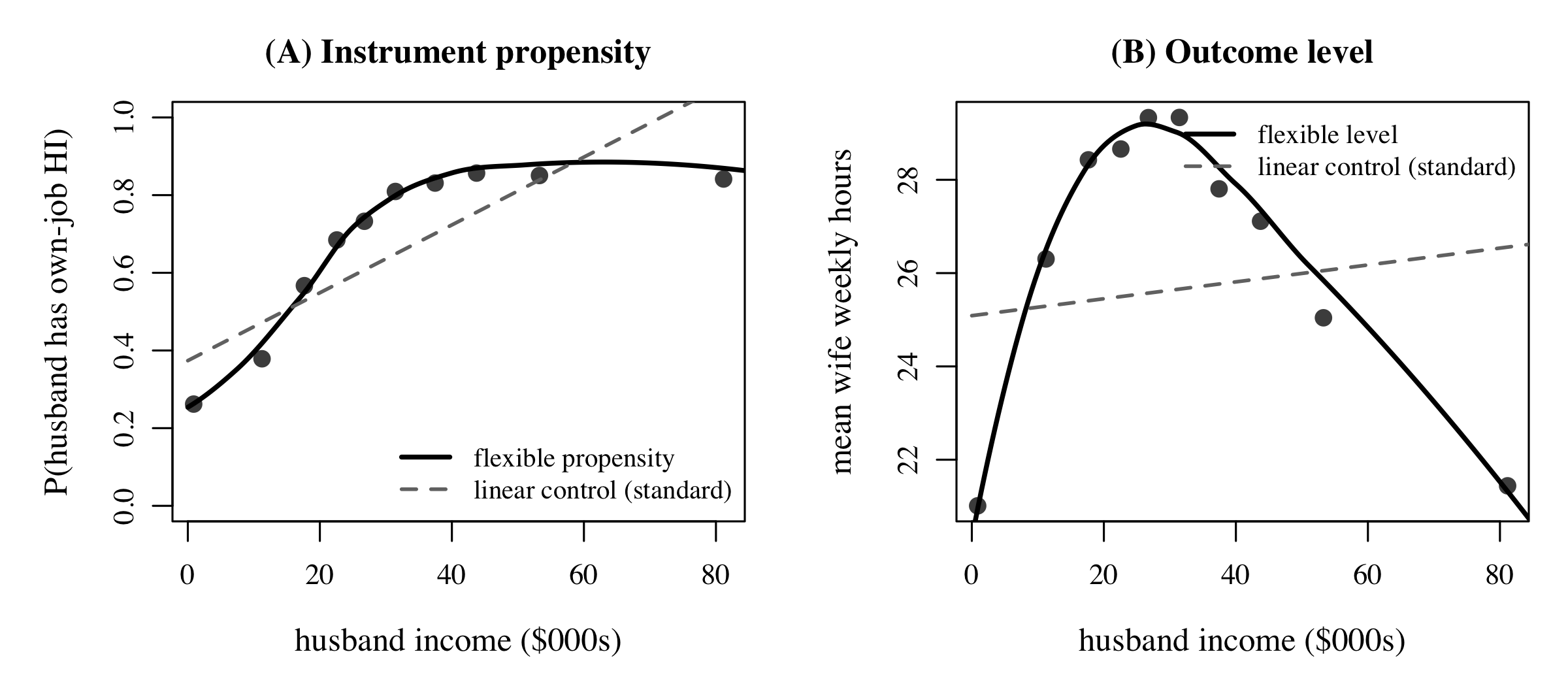}
\caption{Husband income enters both the instrument propensity and the outcome
nonlinearly, while the standard specification enters it linearly. Panel~A, the
probability that the husband holds own-job insurance (the instrument) against husband
income, with binned means, a flexible fit, and the linear control the standard
specification imposes. Panel~B, mean wife weekly hours (the outcome) against husband
income. The propensity curvature a linear control leaves behind (panel~A) covaries with
the curved outcome level (panel~B), and that covariance is the outcome-side
contamination $\theta_Y$ no first-stage number sees.}
\label{fig:hi}
\end{figure}

Table~\ref{tab:hi} reports the just-identified 2SLS estimate under a linear
husband income control, with the diagnostic.\footnote{Treatment $D$ is coverage by the
husband's health insurance, instrument $Z$ is the husband holding insurance through his
own job, outcome $Y$ is the wife's weekly hours worked. The linear control vector is
(education, potential experience, husband income, children under six, children six to
eighteen, race, Hispanic origin, region). The curvature basis $b(X)$ for the naive
specification is a cubic in husband income $(\mathrm{inc}^2,\mathrm{inc}^3)$
residualized against the linear controls, and the flexible benchmark controls husband
income by a natural spline with six degrees of freedom. Inference is a pairs bootstrap
with $B=999$ resamples that refits the propensity, the basis, and the level functions on
each resample. Data ship with the \texttt{Ecdat} package, code in \texttt{hi\_full.R}.}
Under a linear control, the common default, spousal coverage lowers weekly hours by
$1.87$. Both the screening and the bias-targeting tests reject decisively
($p<10^{-3}$) \emph{despite} $F>\num{36000}$, and the bias-corrected estimand is
$-4.45$ ($[-5.14,-3.81]$), more than double the naive figure with an interval that
excludes it. The correction is validated the way it is constructed. Controlling husband
income by a natural spline in an ordinary 2SLS returns $-4.51$, which the correction
from the linear specification reproduces to the second decimal. The finding is
invariant to the basis, moving from $-3.97$ under a quadratic control to $-4.51$ under a
six-knot spline, and both tests reject at every basis.

\begin{table}[t]\centering
\caption{Spousal health insurance and wives' weekly hours ($n=\num{22272}$). Effect of
coverage by the husband's insurance, instrumented by the husband holding own-job
insurance. A conditional $F$ above \num{36000} does not prevent contamination under a
linear husband income control, and the correction recovers the flexible-income answer
without loss of identification. Test $p$-values from the pairs bootstrap. For the
corrected row, cond.\ $F$ is the corrected estimator's own first stage (regressing $D$
on the residualized instrument $Z-\hat e(X)$) and the bracketed interval is the pairs
bootstrap.}
\label{tab:hi}
\resizebox{\textwidth}{!}{%
\begin{tabular}{lccccc}
\toprule
husband income control & 2SLS effect & cond.\ $F$ & $\hat\varrho$ & Directed-S $p$ & Directed-X $p$\\
\midrule
linear (naive) & $-1.87$ & 36301 & 0.09 & $<0.0001$ & $<0.0001$\\
natural spline (benchmark) & $-4.51$ & n/a & n/a & n/a & n/a\\
\emph{contamination-corrected} & \emph{$-4.45$ $[-5.14,-3.81]$} & \emph{31314} & & & \\
\bottomrule
\end{tabular}}
\end{table}

This is the paper's thesis on real data, in near-pure mode~(b). The first stage is
honest and immense, and the contamination share is only $\hat\varrho=0.09$, so both
numbers a first stage reports, the conditional $F>\num{36000}$ and the small share,
describe a clean, strongly identified design. They miss the bias because it loads on
the outcome side, through $\theta_Y$ rather than the first-stage term $\theta_D$. By
Proposition~\ref{prop:impossible} no first-stage diagnostic could reveal this, and none
does. A propensity RESET would register the income curvature as present, as the
screening test does, but presence is not harm, and only the outcome-directed test and
the correction establish that the curvature reaches the estimate here and return the
specification-robust answer of $-4.45$. The design in which every first-stage number is
reassuring is the one that is biased.

Unlike a design in which correction destroys identification
(Proposition~\ref{prop:tradeoff}), here that caveat does not bind. Because the
contamination share is small, the corrected estimator retains almost all of the
first-stage covariance, and its own first-stage $F$, regressing $D$ on the residualized
instrument $Z-\hat e(X)$, is \num{31314}. The corrected estimate is therefore strongly
identified and its bootstrap interval is valid. This is the same eligibility structure
as the 401(k) design below, a treatment attainable only where an instrument grants
eligibility, but it is the cleaner illustration for two reasons. The linear control at
issue is husband income, which a labour-supply equation enters linearly as a matter of
course rather than a specification a careful analyst would already avoid, and
\citet{olson1998} set the design up precisely to compare parametric and semiparametric
estimates, having suspected the functional form that the diagnostic now localizes and
tests against the outcome.

\subsection{A second binary instance: the 401(k) design}\label{sec:401k}

The 401(k) eligibility design reproduces the pattern in the most studied
strong-instrument application in the literature. The effect of 401(k) participation on
household net financial assets is identified using 401(k) \emph{eligibility} as an
instrument \citep{poterba1995,abadie2003}, on the cross-section of $n=9{,}275$
households of \citet{wooldridge2010}. Eligibility is overwhelming, with a conditional
first-stage $F$ exceeding $7{,}000$, and firms offer it in a manner strongly nonlinear
in income while income is a first-order nonlinear determinant of wealth, the
configuration of Lemma~\ref{lem:decomp}.

Table~\ref{tab:401k} reports the just-identified 2SLS estimate (participation on
net financial assets, in thousands of dollars) under three income
specifications, with the diagnostic.\footnote{Full specification, from the released
\texttt{iv\_401k.R}. Treatment $D$ is 401(k) participation, instrument $Z$ is 401(k)
eligibility, outcome $Y$ is net financial assets. The linear control vector is
(income, age, married, male, family size). The curvature basis $b(X)$ for the linear
specification is $(\mathrm{inc}^2,\mathrm{inc}^3,\mathrm{age}^2,\mathrm{inc}\times
\mathrm{age})$ residualized against the linear controls. The quadratic specification
adds $\mathrm{inc}^2,\mathrm{age}^2$ to the controls with basis
$(\mathrm{inc}^3,\mathrm{inc}^4,\mathrm{age}^3,\mathrm{inc}\times\mathrm{age})$. The
spline robustness check uses a natural-spline income basis with six degrees of
freedom, and the flexible benchmark controls income and age by natural splines with
five and four degrees of freedom. Inference is a pairs bootstrap with $B=1499$
resamples that refits the propensity, the curvature basis, and the level functions on
each resample.} Under a linear income control, the common default, the estimate is
$8.40$. Both directed tests reject decisively ($p<10^{-3}$) \emph{despite} $F>7{,}000$,
and the correction returns $12.6$ ($[9.0,16.0]$), reproduced by a quadratic income
control ($13.7$, which also clears the flag at $p_X=0.11$) and a natural-spline control
($13.1$). The naive linear-income estimate is the outlier, understating the effect by
roughly a third. The contamination share is $\hat\varrho=0.03$, so the displacement is
almost entirely the outcome-side term $\theta_Y$, and the corrected estimator's own
first stage, $\num{7295}$, confirms that de-contamination costs no identification here.

\begin{table}[t]\centering
\caption{The 401(k) design ($n=9{,}275$). Effect of participation on net financial
assets (\$000s), instrumented by eligibility. A first-stage $F>7{,}000$ does not
prevent contamination under linear income control, and the correction recovers the
flexible-income answer. Test $p$-values from the pairs bootstrap. For the corrected
row, cond.\ $F$ is the corrected estimator's own first stage (regressing $D$ on the
residualized instrument $Z-\hat e(X)$), the bracketed interval is the pairs
bootstrap, and its identification-robust Anderson--Rubin $95\%$ interval is
$[8.2,17.0]$.}
\label{tab:401k}
\resizebox{\textwidth}{!}{%
\begin{tabular}{lccccc}
\toprule
income control & 2SLS effect & cond.\ $F$ & $\hat\varrho$ & Directed-S $p$ & Directed-X $p$\\
\midrule
linear (naive) & 8.40 & 7401 & 0.03 & $<0.0001$ & $0.0001$\\
quadratic & 13.66 & 7423 & 0.00 & 0.121 & 0.113\\
natural splines (benchmark) & 13.05 & n/a & n/a & n/a & n/a\\
\emph{contamination-corrected} & \emph{12.6 [9.0, 16.0]} & \emph{7295} & & & \\
\bottomrule
\end{tabular}}
\end{table}

Here too a first-stage $F$ above seven thousand offers no protection, and by
Proposition~\ref{prop:impossible} none could. The corrected estimand's
identification-robust Anderson--Rubin $95\%$ interval, $[8.2,17.0]$, coincides with the
normal-theory one, so the corrected estimate is strongly identified, and controlling
income flexibly is exactly what the correction reproduces.

\subsection{Manufactured strength: the slave trade instrument}\label{sec:nw}

Mode~(a) also has a real data face. \citet{nunn2011} instrument an ethnic group's
historical slave exports with its distance from the coast to estimate the effect of the
slave trade on present-day trust, one of the most cited designs in the field, on
$n=\num{17371}$ individuals in $156$ ethnic-group clusters. The instrument looks
powerful, with an unclustered first-stage $F$ above \num{1200}. Clustering at the
ethnic-group level at which it varies already cuts this to $17$, and the diagnostic
explains why even the honest number is generous. The contamination share is
$\hat\varrho=0.52\,[0.13,1.00]$. About half of the residualized instrument is smooth
geographic variation, the curvature of latitude and longitude that distance-to-coast
inherits, rather than identifying variation. Entering those coordinates flexibly rather
than linearly, as \citet{nunn2011} do, moves the trust-in-neighbours estimate from
$-0.243$ to $-0.190$ and attenuates the related trust outcomes by as much as a half,
and it collapses the corrected instrument's own first stage from $17$ to $2.4$. This is
manufactured strength in the sense of Proposition~\ref{prop:collin}. Much of the
apparent identification was the covariate curvature the conditional $F$ cannot separate
from signal, and removing it leaves too little variation to identify the effect. The
bias-targeting test does not reject here, as Proposition~\ref{prop:weakrobust} leads one
to expect under weak identification, since once the manufactured component is stripped
the design no longer speaks, while the screening test, which needs no identification,
still registers the curvature. The instrument's apparent strength is thus in material
part geographic functional form, invisible to the first stage and visible to the share.
Because distance-to-coast is nearly a deterministic function of the coordinates,
controlling them flexibly absorbs much of any geographic instrument, so what the design
establishes is the fragility of the reported strength to how geography enters, not the
magnitude of the effect itself. Appendix Table~\ref{tab:nwrobust} shows the collapse does not depend on the cubic basis. Entering the coordinates quadratically, cubically, or through natural splines leaves the contamination share near one half and drives the corrected first stage to between $2.4$ and $3.0$ in every case, for trust in relatives as well as trust in neighbours, while the cluster-robust first stage of the original instrument stays at $17.0$ throughout.

Taken together the three applications exhibit both failure modes on canonical data. The
health insurance and 401(k) designs are mode~(b), where an honest first stage, above
thirty-six thousand and above seven thousand, conceals an outcome-side bias the directed
test detects and the correction repairs. The slave trade design is mode~(a), where the
apparent strength is in material part covariate curvature and correcting it collapses
the first stage. In each the conditional $F$ is silent, whether because the bias loads
on the outcome or because the strength itself is manufactured, and the directed test and
the reportable share speak where it cannot.

%=============================================================================
\section{Recommendations for practice}\label{sec:practice}
%=============================================================================

Our results imply a short protocol that adds little cost to a standard IV
analysis. (1) Report the \emph{conditional} first-stage $F$ of
\citet{sanderson2016}, but do not treat it as evidence for a causal reading of
the second stage, and report the estimated contamination share
$\hat\varrho$~\eqref{eq:share} beside it. (2) Estimate the instrument propensity
$\E[Z\mid X]$ flexibly and test its curvature against $D$ and $Y$ using the
screening statistic of Section~\ref{sec:test} and, to distinguish harmful from
harmless curvature, the bias-targeting test of Section~\ref{sec:exact}. A screening
rejection signals curvature that loads on $D$ and $Y$ even when the conditional $F$
is large, but it does not by itself establish bias. Use the bias-targeting test to
decide, reading its signal-cleared form $A_n$~\eqref{eq:arstat} for a size guarantee
that holds even when the design is weakly identified and the plug-in $T_n^{\ast}$ for
its extra power when it is not. (3) If the \emph{bias-targeting} test rejects, report
the bias-corrected estimand~\eqref{eq:correction} together with weak-IV-robust
inference, and
recognize (Proposition~\ref{prop:tradeoff}) that a collapse in the corrected
first stage means the apparent strength was contamination rather than
identification. A screening rejection that the bias-targeting test does not confirm
is the harmless-curvature case and calls for no correction. (4) Where feasible, prefer an instrument whose propensity is
approximately linear in the controls, or saturate the controls at the design
stage while monitoring the resulting conditional strength.

%=============================================================================
\section{Conclusion}\label{sec:conclusion}
%=============================================================================

The applied ritual of reporting a first-stage $F$ answers whether an instrument
moves the treatment, not whether the second-stage estimate is a meaningful causal
average. Under collinearity between instrument and covariates these questions come apart in a
specific, detectable way. The same propensity curvature that contaminates the
2SLS estimand also inflates the conditional $F$, so the strength
diagnostic cannot see the problem and the natural saturation fix trades the bias
for weak identification. A directed conditional-moment test, built entirely from
reduced-form regressions, has power exactly where the conditional $F$ is silent.

The analysis covers a scalar instrument and a scalar treatment of arbitrary support,
entered in a just-identified linear model with linear covariates, with the
binary-treatment LATE model as the leading special case. We now say what survives
beyond that boundary and what does not. Over-identification is handled in
Section~\ref{sec:overid}, the less mechanical direction. With a vector of instruments the
numerator and denominator of Lemma~\ref{lem:decomp} become signal-plus-contamination
vectors and 2SLS combines them through the residualized-instrument second-moment matrix, so
the population coefficient is a weighted combination rather than a clean term-by-term ratio.
The stacked screen $\theta_D=\theta_Y=0$ stays sufficient for no contamination, the
impossibility result (Proposition~\ref{prop:impossible}) holds verbatim since the bias
still loads on the outcome, and the exact bias-targeting restriction, which inverts that
second-moment matrix, is the single moment \eqref{eq:oidexact}, confirmed against the
stacked screen in Section~\ref{sec:overidMC}. Multiple, that is \emph{vector},
treatments are the substantive frontier. A multivalued single treatment is already
covered above, but with a vector of treatments our single-treatment functional-form
contamination compounds with the \emph{cross-treatment} contamination of
\citet{goldsmith2024}, a conceptually distinct channel that persists even under a
saturated propensity. Disentangling the two, where the object to be tested is a matrix
rather than a scalar contamination, is the main extension we leave to future work. What
extends cleanly in every case is the reason the first stage cannot help. The bias always
loads on the conditional law of the outcome, so Proposition~\ref{prop:impossible} is the
general fact, of which the binary model is one instance. The remaining loose end, a
data-driven curvature basis, is closed by Proposition~\ref{prop:orthogonal}.

\clearpage
%=============================================================================
\appendix
\section{Proofs}\label{app:proofs}
%=============================================================================

\subsection*{Proof of Lemma~\ref{lem:decomp}}
By \eqref{eq:fwl}, $\beta_{2SLS}=\E[\Ztil Y]/\E[\Ztil D]$ with $\Ztil=Z-\ell(X)$. Write
$\Ztil=(Z-e(X))+h(X)$ using $h=e-\ell$. For any conditionally square-integrable $W\in\{D,Y\}$,
\[
\E[\Ztil W]=\E[(Z-e(X))W]+\E[h(X)W].
\]
For the first term, condition on $X$. Since $\E[Z-e(X)\mid X]=0$,
$\E[(Z-e(X))W\mid X]=\E[(Z-e(X))(W-m_W(X))\mid X]=\Cov(Z,W\mid X)$, and taking expectations over $X$
gives $S_W$. For the second term, $\E[h(X)W]=\E[h(X)\,\E[W\mid X]]=\E[h(X)m_W(X)]=\Cov(h(X),m_W(X))$,
the last equality because $\E[h(X)]=0$. This gives $\theta_W$ and \eqref{eq:decomp}.
Equation~\eqref{eq:bias} follows by subtracting $\bar L=S_Y/S_D$ from \eqref{eq:decomp} and
simplifying, and $S_D+\theta_D=\E[\Ztil D]$ by construction. The binary specialization
(Remark~\ref{rem:binary}) follows because for binary $Z$,
$\Cov(Z,W\mid X)=e(X)\{1-e(X)\}\{\E[W\mid Z=1,X]-\E[W\mid Z=0,X]\}$, which is $e(1-e)p$ for $W=D$ and,
by the conditional Wald identity under the LATE assumptions, $e(1-e)p\,\late$ for $W=Y$. \hfill$\qed$

\subsection*{Proof of Proposition~\ref{prop:collin}}
By Assumption~\ref{as:seq}(i), $\Var_t(Z\mid X)\to0$ a.e.\ and in $L^1$. By Cauchy--Schwarz,
$|\Cov(Z,W\mid X)|\le\sqrt{\Var_t(Z\mid X)\,\Var_t(W\mid X)}$. For $W=D$, with
$\sup_t\E[\Var_t(D\mid X)]<\infty$, the Cauchy--Schwarz inequality applied to the outer expectation
gives $|S_{D,t}|\le\sqrt{\E[\Var_t(Z\mid X)]\,\E[\Var_t(D\mid X)]}\to0$. For $W=Y$, the integrand
$|\Cov_t(Z,Y\mid X)|\le\sqrt{\Var_t(Z\mid X)}\sqrt{\Var_t(Y\mid X)}\to0$ a.e., and the family is
uniformly integrable because $\{\Var_t(Y\mid X)\}$ is and $\Var_t(Z\mid X)$ is bounded, so Vitali's
theorem gives $S_{Y,t}\to0$. By Assumption~\ref{as:seq}(ii)--(iii), $h_t\to h^\ast$ and
$m_{D,t},m_{Y,t}$ converge in $L^2$, so by Cauchy--Schwarz $\theta_{D,t}\to\Cov(h^\ast,m_D^\ast)$ and
$\theta_{Y,t}\to\Cov(h^\ast,m_Y^\ast)$. The denominator converges to $\Cov(h^\ast,m_D^\ast)\neq0$ by
Assumption~\ref{as:seq}(iv), so \eqref{eq:limit} follows from the continuous mapping theorem. That the
limit can fall outside the hull of within-cell slopes, and in the binary specialization outside the
LATE hull, is confirmed by the binary design of Section~\ref{sec:mcbinary}
(Table~\ref{tab:collinbinary}, Figure~\ref{fig:collinbinary}), where every conditional LATE is
positive and bounded by $1.5$ yet the estimand converges to a limit near $9$. \hfill$\qed$

\subsection*{Proof of Corollary~\ref{cor:Ffails}}
The population conditional first-stage slope from regressing $D$ on $\Ztil$ is
$\pi_t=\E[\Ztil_t D]/\Var(\Ztil_t)$. By Lemma~\ref{lem:decomp},
$\E[\Ztil_t D]=S_{D,t}+\theta_{D,t}\to\Cov(h^\ast,m_D^\ast)$, and
$\Var(\Ztil_t)=\E[\Var_t(Z\mid X)]+\E[h_t^2]\to\E[(h^\ast)^2]$ by Assumption~\ref{as:seq}(i)--(ii),
using $\Var(\Ztil_t)=\E[\Var_t(Z\mid X)]+\Var(h_t(X))$ and $\E[h_t]=0$. Hence
$\pi_t\to\Cov(h^\ast,m_D^\ast)/\E[(h^\ast)^2]\neq0$. The conditional $F$ has population analogue
$n\,\pi_t^2\,\Var(\Ztil_t)/\sigma^2_v$ with $\sigma^2_v=\Var(D-\pi_t\Ztil-X'c)$ bounded, so it is
bounded away from zero and diverges in $n$. \hfill$\qed$

\subsection*{Proof of Proposition~\ref{prop:impossible}}
Fix the joint law of $(Z,D,X)$ with $h\neq0$. Then $e(\cdot)$, $\ell(\cdot)$, $h(\cdot)$,
$m_D(\cdot)$, $S_D=\E[\Cov(Z,D\mid X)]$, $\theta_D=\E[h(X)D]$ and $\Var(\Ztil)$ are all determined, so
any functional $T$ of this law takes a fixed value. By \eqref{eq:bias},
$\beta_{2SLS}-\bar L=(\theta_Y-\bar L\theta_D)/(S_D+\theta_D)$ with $\theta_Y=\E[h(X)Y]$ and
$\bar L=S_Y/S_D$, $S_Y=\E[\Cov(Z,Y\mid X)]$. Both $\theta_Y$ and $S_Y$ are determined by the
conditional law of $Y$ given $(Z,X)$, which is not restricted by the marginal law of $(Z,D,X)$.
Construct two outcome laws sharing that marginal but with conditional means $m_Y$ and $m_Y'$ that
differ on the span orthogonal to $(1,X')$, so that $\Cov(h,m_Y)\neq\Cov(h,m_Y')$. Because $h\neq0$
such a perturbation exists. The perturbation changes only the level $\E[Y\mid X]$ and leaves the law
of $(Z,D,X)$ and the within-cell slopes untouched, so both DGPs are admissible. The two biases then
differ while $T$ is identical, and since the gap $\Cov(h,m_Y)-\Cov(h,m_Y')$ can be made arbitrarily
large, the bias is unbounded as a function of the outcome law at fixed $T$. Hence $T$ carries no
information about the bias. When $h\equiv0$ the construction is empty because $\Cov(h,\cdot)\equiv0$,
consistent with the first stage then certifying no contamination. \hfill$\qed$

\subsection*{Proof of Proposition~\ref{prop:test}}
$\hat\theta$ is a two-step estimator, a first-step OLS $\hat\delta$ from
regressing $Z$ on $(1,X',b(X)')$, and second-step sample means of $\hat h_i W_i$,
$W\in\{D,Y\}$. Under the stated moment conditions and fixed $K$,
$\hat\delta\convp\delta$ and is asymptotically linear with influence
$G^{-1}b(X_i)u_i$, $G=\E[bb']$. A mean-value expansion of
$\hat\theta_W=n^{-1}\sum_i b(X_i)'\hat\delta\,W_i$ around $\delta$ gives
$\sqrt n(\hat\theta_W-\theta_W)=n^{-1/2}\sum_i\{b(X_i)'\delta\,W_i-\theta_W\}
+\E[W b(X)']\sqrt n(\hat\delta-\delta)+o_p(1)$. Stacking $W\in\{D,Y\}$ and substituting
the first-step influence gives the influence function
\[
\psi_i=\begin{pmatrix} b(X_i)'\delta\,D_i-\theta_D\\ b(X_i)'\delta\,Y_i-\theta_Y\end{pmatrix}
+\Big(\E[D\,b(X)'],\ \E[Y\,b(X)']\Big)'\,G^{-1}b(X_i)u_i,
\qquad u_i=Z_i-\alpha_0-X_i'\alpha_1-b(X_i)'\delta.
\]
The central limit theorem yields
$\sqrt n(\hat\theta-\theta)\convd N(0,\Sigma)$ with
$\Sigma=\Var(\psi_i)$. Under $H_0$, $\theta=0$ and $W_n=n\hat\theta'\hat\Sigma^{-1}
\hat\theta\convd\chi^2_2$ provided $\hat\Sigma\convp\Sigma\succ0$. Consistency
against fixed alternatives and the local power statement are immediate from
$\hat\theta\convp\theta\neq0$ and $\theta=\eta/\sqrt n$, respectively.
Consistency of the pairs bootstrap for $\Sigma$ follows from Hadamard
differentiability of $\hat\theta$ in the empirical measure. \hfill$\qed$

\subsection*{Proof of Proposition~\ref{prop:weakrobust}}
Fix a sequence $\{P_n\}$ as stated. The estimator $\hat\theta$ and its influence
function
$\psi_{i}=(b(X_i)'\delta\,W_i-\theta_W)_{W\in\{D,Y\}}+\big(\E[Wb(X)']\big)_W G^{-1}b(X_i)u_i$
from the proof of Proposition~\ref{prop:test} involve the first-stage strength
nowhere. They are determined by the reduced-form regression of $Z$ on
$(1,X',b(X)')$ and by the cross-moments $\E[Wb(X)']$, none of which is the
first-stage covariance $\E[\Ztil D]$. The weak-instrument problem is a
small-denominator problem for $\beta_{2SLS}=\E[\Ztil Y]/\E[\Ztil D]$, and that
denominator does not appear. The uniform fourth-moment bound and
$\Sigma_n\to\Sigma_\ast\succ0$ posited for the null sequence supply a Lyapunov
condition, so the Lindeberg--Feller central limit theorem for triangular arrays
applies to $n^{-1/2}\sum_i\psi_{i,n}$, giving
$\sqrt n(\hat\theta-\theta)\convd N(0,\Sigma_\ast)$ along the sequence. Under $H_0$,
$\theta=0$ and $\hat\Sigma\convp\Sigma_\ast$, so $W_n\convd\chi^2_2$. The
generated-regressor term $\E[Wb(X)']G^{-1}b(X_i)u_i$ has variance of order
$\Var(u_i)\to0$ as $Z$ becomes collinear, but the leading term
$b(X_i)'\delta\,W_i-\theta_W$ has variance bounded below because the curvature stays
nondegenerate, $h_n\to h^\ast\neq0$ and hence $\delta\to\delta^\ast\neq0$. This
delivers $\Sigma_\ast\succ0$. The identical argument applied to the orthogonalized
influence $\psi^o_i$, which is likewise free of the first-stage covariance, gives
$W^o_n\convd\chi^2_2$. Finally, along the contaminated sequence of
Assumption~\ref{as:seq} the same central limit theorem gives
$\sqrt n(\hat\theta-\theta_n)\convd N(0,\Sigma_\ast)$ while $\theta_n\to\theta^\ast$
with $\theta_D^\ast=\Cov(h^\ast,m_D^\ast)\neq0$. The noncentrality
$n\,\theta_n'\Sigma_n^{-1}\theta_n\to\infty$, so $W_n\to\infty$ in probability and
the power tends to one. \hfill$\qed$

\subsection*{Proof of Proposition~\ref{prop:orthogonal}}
Mean zero at the truth is immediate, because
$\E[\psi^o_W]=\E[h(X)W]+\E[\tilde m_W(X)(Z-e(X))]-\theta_W=\theta_W+0-\theta_W=0$,
using $\E[Z-e(X)\mid X]=0$. For orthogonality with respect to $e$, perturb
$e\to e+t\Delta$. Because $h=e-L(e\mid1,X)$ is linear in $e$ and the linear
projection is self-adjoint,
$\partial_t\E[h(X)W]\big|_{t=0}=\E[(\Delta-L(\Delta\mid1,X))\,m_W(X)]
=\E[\Delta\,\tilde m_W(X)]$, while
$\partial_t\E[\tilde m_W(X)(Z-e(X))]\big|_{t=0}=-\E[\tilde m_W(X)\,\Delta(X)]$. The
two terms cancel. Perturbing $m_W\to m_W+t\,\Xi$ moves only the correction term,
with derivative $\E[\tilde\Xi(X)(Z-e(X))]=\E[\tilde\Xi(X)\,\E(Z-e(X)\mid X)]=0$.
Hence $\psi^o_W$ is Neyman-orthogonal in both nuisances. The affine projection
$L(\cdot\mid1,X)$ is not a third nuisance but a linear functional of $e$. The
$e$-perturbation above moves $h=e-L(e\mid1,X)$ through the projection, and the same
self-adjointness that makes the two derivatives cancel annihilates its first-order
contribution \emph{to the pathwise derivative}. In the estimator the projection is a
parametric $\sqrt n$ step, an OLS of $Z$ on $(1,X)$. Its estimation error contributes
a regular $O_p(n^{-1/2})$ term to the influence function, not orthogonalized away but
asymptotically linear, which the pairs bootstrap captures by recomputing the
projection, the curvature basis, and the level functions on each resample. With $K$-fold cross-fitting the
plug-in decomposes as
$\hat\theta^o_W-\theta_W=n^{-1}\sum_i\psi^o_W(O_i;e,m_W)
+\text{(bias)}+\text{(remainder)}$. Orthogonality bounds the bias by the product
$\|\hat e-e\|_{L^2}\|\hat m_W-m_W\|_{L^2}=o_p(n^{-1/2})$, and cross-fitting makes
the remainder $o_p(n^{-1/2})$ under the stated $L^2$ rates and fourth moments
\citep{chernozhukov2018}. The leading term is a sample average of the fixed
influence $\psi^o_i$, so $\sqrt n(\hat\theta^o-\theta)\convd N(0,\Sigma_o)$ by the
central limit theorem, and $W^o_n\convd\chi^2_2$ under $H_0$ follows as in the
proof of Proposition~\ref{prop:test}. \hfill$\qed$

\subsection*{Proof of Propositions~\ref{prop:correction} and \ref{prop:tradeoff}}
Write the numerator of \eqref{eq:correction} as
$n^{-1}\sum_i(Z_i-\hat e(X_i))Y_i=n^{-1}\sum_i(Z_i-e(X_i))Y_i
-n^{-1}\sum_i(\hat e(X_i)-e(X_i))Y_i$. The first term converges in probability to
$\E[(Z-e(X))Y]=S_Y$ by the argument in the proof of Lemma~\ref{lem:decomp}. For
the second, by Cauchy--Schwarz its magnitude is at most
$(n^{-1}\sum_i(\hat e(X_i)-e(X_i))^2)^{1/2}(n^{-1}\sum_i Y_i^2)^{1/2}
=\|\hat e-e\|_{n}\cdot O_p(1)$, where $\|\hat e-e\|_{n}^2=n^{-1}\sum_i(\hat
e(X_i)-e(X_i))^2$ is the in-sample empirical norm. This is $o_p(1)$ whenever
$\|\hat e-e\|_{n}=o_p(1)$. For the fixed, low-dimensional basis used throughout,
the empirical norm converges to the population norm $\|\hat e-e\|_{L^2}$ by a
uniform law of large numbers, so mere $L^2$-consistency of the propensity estimator
suffices, with no orthogonality or rate requirement. For a flexible,
high-dimensional or machine-learning $\hat e$ trained in-sample, population
$L^2$-consistency need not control the empirical norm, and one should then cross-fit
$\hat e$, evaluating it on observations not used to fit it, so that the displayed
average is out-of-sample and the argument goes through unchanged. The
denominator converges
to $\E[(Z-e(X))D]=S_D$ identically, and provided $S_D\neq0$ in the fixed
data-generating process the ratio converges to $S_Y/S_D=\bar L$, proving
Proposition~\ref{prop:correction}. (Asymptotic normality and inference for
$\hat{\bar L}$, as opposed to consistency, do require controlling the first-step
influence, for example a Neyman-orthogonalized moment with cross-fitting, and,
because $S_D$ may be small, weak-instrument-robust confidence sets. See
Proposition~\ref{prop:orthogonal} and Section~\ref{sec:correction}.) For
Proposition~\ref{prop:tradeoff}, the corrected estimator's first-stage covariance
is exactly this denominator $\E[(Z-e(X))D]=S_D=\E[e(X)\{1-e(X)\}p(X)]$, which by
Proposition~\ref{prop:collin} tends to $0$ under Assumption~\ref{as:seq}. The
identity of Corollary~\ref{cor:frontier} is immediate. Since $\varrho=\theta_D/\E[\Ztil D]$
by \eqref{eq:share} and $S_D=\E[\Ztil D]-\theta_D$, we have $S_D=(1-\varrho)\E[\Ztil D]$.
The concentration-parameter statement follows because the corrected instrument
$Z-e(X)$ has covariance $S_D$ with $D$ and variance $\E[\Var(Z\mid X)]$. \hfill$\qed$

\subsection*{Asymptotics of the bias-targeting statistic $T_n^{\ast}$}
Let $\hat g=n^{-1}\sum_i\hat h_i(Y_i-\hat{\bar L}D_i)$ estimate
$g(\bar L)=\E[h(X)(Y-\bar L D)]$, which equals $\theta_Y-\bar L\theta_D$ and is
zero under $H_0^{\ast}$ by \eqref{eq:bias}. Expanding in $(\hat\delta,\hat{\bar L})$
about $(\delta,\bar L)$,
$\sqrt n\,\hat g=n^{-1/2}\sum_i\{h(X_i)(Y_i-\bar L D_i)-g(\bar L)\}
+\E[b(X)'(Y-\bar L D)]\,\sqrt n(\hat\delta-\delta)
-\theta_D\,\sqrt n(\hat{\bar L}-\bar L)+o_p(1)$.
We require, beyond Assumption~\ref{as:test} and $S_D$ bounded away from zero, that
$\hat{\bar L}$ be $\sqrt n$-asymptotically linear with an influence function
$\phi_i$, $\sqrt n(\hat{\bar L}-\bar L)=n^{-1/2}\sum_i\phi_i+o_p(1)$. This is a
genuine additional condition. Proposition~\ref{prop:correction} delivers only
consistency of $\hat{\bar L}$, so its asymptotic linearity must be supplied
separately, for instance by the Neyman-orthogonalized, cross-fitted estimator of
Proposition~\ref{prop:orthogonal} under those rate
conditions, which is why we invoke that construction rather than the plug-in when we
need an interval. Given it, $\sqrt n(\hat\delta-\delta)$ and
$\sqrt n(\hat{\bar L}-\bar L)$ are both asymptotically linear, so
$\sqrt n\,\hat g\convd N(0,v)$ for a finite $v>0$. In this fixed, strongly identified
data-generating process the pairs bootstrap is first-order valid for the smooth
functional $\hat g$ by the same Hadamard-differentiability argument as in
Proposition~\ref{prop:test}, so the bootstrap variance $\hat v$ satisfies
$n\hat v\convp v$. Then
$T_n^{\ast}=\hat g^2/\hat v=(\sqrt n\,\hat g)^2/(n\hat v)\convd\chi^2_1$
under $H_0^{\ast}$, and $\hat g\convp g(\bar L)\neq0$ gives consistency against any
fixed alternative with nonzero bias. (Equivalently, one may write the statistic as
$n\hat g^2/\hat v_1$ with $\hat v_1\convp v$ a consistent estimate of the asymptotic
variance of $\sqrt n\,\hat g$, and the two forms coincide since $\hat v_1=n\hat v$.) As
$S_D\to0$ the term $-\theta_D\sqrt n(\hat{\bar L}-\bar L)$ ceases to be $O_p(1)$,
$\hat{\bar L}$ is no longer $\sqrt n$-consistent, and both the Gaussian limit and the
bootstrap approximation for $T_n^{\ast}$ fail. Size control at that boundary is
recovered not by $T_n^{\ast}$ but by the signal-cleared statistic $A_n$ of
\eqref{eq:arstat}, whose validity along the collinearity sequence is established in
Proposition~\ref{prop:arbias} below, which is why the protocol reads $A_n$ for size
and reserves $T_n^{\ast}$ for the extra power it buys under strong identification.
\hfill$\qed$

\subsection*{Proof of Proposition~\ref{prop:arbias}}
Write $m=(\theta_D,\theta_Y,S_D,S_Y)'$ and $\hat m$ for its plug-in from the flexible
first-step regression, with $\psi=q(m)=\theta_Y S_D-\theta_D S_Y$ and gradient
$\nabla\psi=(-S_Y,S_D,\theta_Y,-\theta_D)'$. Each coordinate of $\hat m$ is a two-step
sample mean of the kind analysed in Propositions~\ref{prop:test}
and~\ref{prop:weakrobust}. For $\hat\theta_W$ the influence function is the
$\psi_i$ of Proposition~\ref{prop:test}, and for $\hat S_W=n^{-1}\sum_i(Z_i-\hat
e(X_i))W_i$ it is $(Z_i-e(X_i))W_i-S_W$ plus the same generated-regressor correction,
so along $\{P_n\}$ the stacked vector obeys
$\sqrt n\,(\hat m-m_n)=n^{-1/2}\sum_i\zeta_{i,n}+o_p(1)$ with
$\Omega_n=\Var(\zeta_{i,n})$, and a Lyapunov condition holds under the uniform fourth
moments of Assumption~\ref{as:test}, exactly as in
Proposition~\ref{prop:weakrobust}. A first-order expansion gives
$\hat\psi-\psi_n=\nabla\psi_n'(\hat m-m_n)+R_n$, where the remainder $R_n$ is a sum of
products of two centered coordinates of $\hat m-m_n$. Since each coordinate is
$O_p(n^{-1/2}\sigma_{\cdot,n})$ with $\sigma_{\cdot,n}$ the sd of the relevant
influence, $\sqrt n\,R_n=O_p(n^{-1/2}\max_k\sigma_{k,n}^2)=o_p(\sigma_{\psi,n})$, where
$\sigma_{\psi,n}^2=\nabla\psi_n'\Omega_n\nabla\psi_n$ is the asymptotic variance of
$\sqrt n\,\hat\psi$, so the remainder is negligible after self-normalization. Under
$H_0^{\ast}$, $\psi_n=0$, hence
$\sqrt n\,\hat\psi/\sigma_{\psi,n}=\nabla\psi_n'n^{-1/2}\sum_i\zeta_{i,n}/\sigma_{\psi,n}+o_p(1)
\convd N(0,1)$
by the Lindeberg--Feller theorem for triangular arrays, provided
$\sigma_{\psi,n}^2$ stays bounded away from zero after normalization, which holds
because as $S_{D,n},S_{Y,n}\to0$ the gradient tends to $(0,0,\theta_Y^\ast,-\theta_D^\ast)$
with $(\theta_D^\ast,\theta_Y^\ast)\neq0$, so
$\sigma_{\psi,n}^2\to\Var(\theta_Y^\ast\zeta^{S_D}-\theta_D^\ast\zeta^{S_Y})>0$ under
the nondegeneracy of the reduced-form scores. This last step uses that the statistic
is self-normalized, so only the ratio $\sqrt n\,\hat\psi/\sigma_{\psi,n}$ matters and
the common vanishing of $\psi$ and its standard error cancels. With $\hat v_\psi$
ratio-consistent, $n\hat v_\psi/\sigma_{\psi,n}^2\convp1$, so
$A_n=\hat\psi^2/\hat v_\psi=(\sqrt n\,\hat\psi/\sigma_{\psi,n})^2\,(\sigma_{\psi,n}^2/n\hat v_\psi)
\convd\chi^2_1$, whether or not $S_{D,n}\to0$. For power, at a fixed law with
$b\neq0$ and $S_D$ bounded away from zero, $\psi=b\,S_D\,\E[\Ztil D]\neq0$ is fixed
while $\hat v_\psi=O_p(n^{-1})$, so $A_n\to\infty$ and the power tends to one. Holding
$b$ fixed while $S_D\to0$ sends $\psi=b\,S_D\,\E[\Ztil D]\to0$, so the noncentrality
$\psi_n^2/\Var(\hat\psi)$ need no longer diverge and the power against that fixed $b$
declines toward $\alpha$, the Anderson--Rubin trade-off. \hfill$\qed$

%=============================================================================
\section{Additional Monte Carlo evidence}\label{app:mc}
%=============================================================================

This appendix collects the size-and-power, power-curve, collinearity-sweep, and
share-estimation exhibits summarized in Section~\ref{sec:mc}. The design is that of
Section~\ref{sec:mc}.

\begin{table}[h]\centering
\caption{Directed contamination test. Panel A, one dataset, $n=\num{100000}$.
Panel B, rejection rates, $5\%$ level, $\num{2000}$ replications, $n=3000$, cubic basis.}
\label{tab:sizepower}
\begin{tabular}{lccccc}
\multicolumn{6}{l}{\emph{Panel A. Single contaminated dataset}}\\
\toprule
$\hat\theta_D$ & $\hat\theta_Y$ & $p$-value & naive $\beta_{2SLS}$ & cond.\ $F$ & corrected $\hat{\bar L}$\\
\midrule
0.067 & 0.800 & $<10^{-4}$ & 8.52 & \num{19737} & 1.03\\
\bottomrule
\end{tabular}

\vspace{1em}
\begin{tabular}{lcc}
\multicolumn{3}{l}{\emph{Panel B. Size and power}}\\
\toprule
design & rejection rate ($5\%$) & mean cond.\ $F$\\
\midrule
SIZE\ \ (linear propensity, $H_0$ true)   & 0.019 & 569\\
POWER (curved propensity, $H_0$ false)     & 1.000 & 606\\
\bottomrule
\end{tabular}
\end{table}

\begin{table}[h]\centering
\caption{Clean design (linear propensity, no bias), Directed-X size against $n$ at the
$5\%$ nominal level, $\num{2000}$ replications, $B=999$. ``Dir-X'' is the plug-in
$T_n^{\ast}$ and ``AR'' the signal-cleared $A_n$. Monte Carlo standard errors are near
$0.005$. Studentizing $T_n^{\ast}$ with too few bootstrap draws ($B=99$) inflates the
clean-design rejection rate through a noisy variance estimate, and at $B=999$ it is at
nominal.}
\label{tab:nsweep}
\begin{tabular}{rccc}
\toprule
$n$ & Dir-X size & AR size & mean cond.\ $F$\\
\midrule
\num{1000}  & 0.058 & 0.058 & 192\\
\num{3000}  & 0.049 & 0.048 & 570\\
\num{9000}  & 0.044 & 0.045 & \num{1704}\\
\num{20000} & 0.044 & 0.045 & \num{3783}\\
\bottomrule
\end{tabular}
\end{table}

The power curve sweeps the linear-probability curvature knob $b$ for three sample
sizes. The top axis of Figure~\ref{fig:power} maps $b$ to the induced population
bias $\beta_{2SLS}-\bar L$. At $b=0$ the propensity is exactly linear and the test
controls size, rejecting at $1.4\%$, $1.7\%$ and $1.3\%$ for $n=1000,3000,9000$
against the $5\%$ nominal level (the $b=0$ row of Table~\ref{tab:power}), so the
pairs bootstrap is mildly conservative rather than exactly calibrated in this
design.
Power then rises monotonically in $b$ and in $n$, while the
mean conditional $F$ stays pinned near $\num{1740}$, so the same $F$ accompanies a
bias of $0$ and a bias of $2.64$.

\begin{figure}[h]\centering
\includegraphics[width=\textwidth]{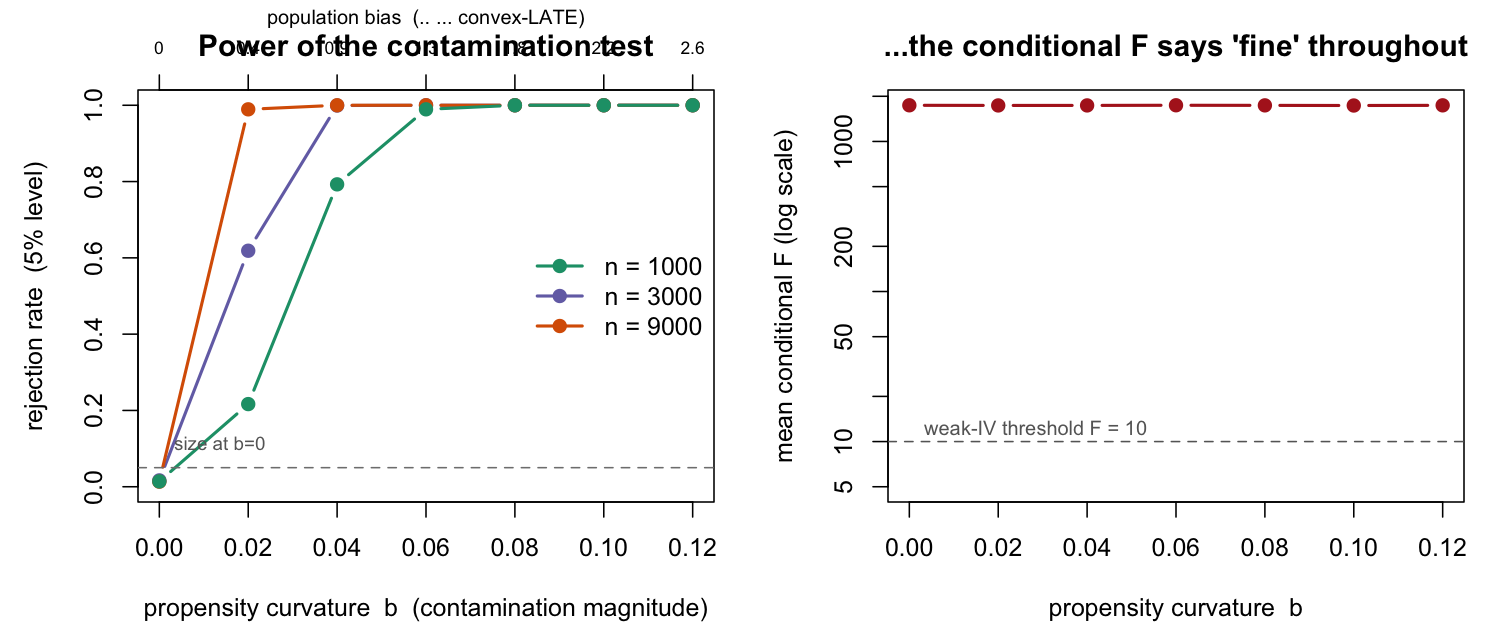}
\caption{Left panel. Rejection rate of the directed test vs.\ propensity
curvature $b$ (bottom axis) and induced population bias (top axis), for
$n\in\{1000,3000,9000\}$, with the dashed line at the $5\%$ level. Right panel.
The mean conditional $F$ (log scale) is invariant to $b$, sitting near
$\num{1740}$ throughout.}
\label{fig:power}
\end{figure}

\begin{table}[h]\centering
\caption{Power curve. Rejection rates ($5\%$ level, $\num{1500}$ replications,
$B=299$) by curvature knob $b$ and sample size. Induced population bias
$\beta_{2SLS}-\bar L$ at $n=\num{200000}$. Mean conditional $F$ at $n=9000$.}
\label{tab:power}
\begin{tabular}{rrrrrr}
\toprule
$b$ & bias & $n=1000$ & $n=3000$ & $n=9000$ & mean cond.\ $F$\\
\midrule
0.00 & $-0.02$ & 0.014 & 0.017 & 0.013 & \num{1742}\\
0.02 & 0.40    & 0.217 & 0.619 & 0.989 & \num{1738}\\
0.04 & 0.87    & 0.793 & 0.999 & 1.000 & \num{1738}\\
0.06 & 1.32    & 0.989 & 1.000 & 1.000 & \num{1743}\\
0.08 & 1.80    & 1.000 & 1.000 & 1.000 & \num{1740}\\
0.10 & 2.18    & 1.000 & 1.000 & 1.000 & \num{1736}\\
0.12 & 2.64    & 1.000 & 1.000 & 1.000 & \num{1741}\\
\bottomrule
\end{tabular}
\end{table}

\begin{table}[h]\centering
\caption{Directed tests as collinearity between instrument and covariates rises, in the
contaminated design (rejection $=$ power, the bias it faces in the third column) and
the harmless design (rejection $=$ size, $5\%$ nominal). $\num{2000}$ replications,
$n=3000$, cubic basis, $B=999$. ``Dir-X'' is the plug-in $T_n^{\ast}$ and ``AR'' the
signal-cleared $A_n$. Directed-S raises a false alarm in every harmless row because it tests
the sufficient condition $\theta=0$. Both bias-targeting forms keep full power against
genuine bias and hold size across the collinearity range, the size guarantee for $A_n$
being that of Proposition~\ref{prop:arbias}.}
\label{tab:svx}
\begin{tabular}{rrccccccc}
\toprule
 & & \multicolumn{4}{c}{Contaminated (power)} & \multicolumn{3}{c}{Harmless (size)}\\
\cmidrule(lr){3-6}\cmidrule(lr){7-9}
$\kappa$ & $\Var(Z\mid X)$ & bias & Dir-S & Dir-X & AR & Dir-S & Dir-X & AR\\
\midrule
0.5  & 0.230 & 2.56 & 1.00 & 1.00 & 1.00 & 1.00 & 0.044 & 0.044\\
1.0  & 0.190 & 4.44 & 1.00 & 1.00 & 1.00 & 1.00 & 0.051 & 0.051\\
2.0  & 0.119 & 6.25 & 1.00 & 1.00 & 1.00 & 1.00 & 0.044 & 0.045\\
4.0  & 0.058 & 7.45 & 1.00 & 1.00 & 1.00 & 1.00 & 0.050 & 0.052\\
8.0  & 0.028 & 7.78 & 1.00 & 1.00 & 1.00 & 1.00 & 0.044 & 0.050\\
16.0 & 0.014 & 7.83 & 1.00 & 1.00 & 1.00 & 1.00 & 0.040 & 0.043\\
\bottomrule
\end{tabular}
\end{table}

\begin{table}[h]\centering
\caption{Robustness of the slave trade finding (Section~\ref{sec:nw}) to the geography
basis and the clustering level. The sample is that of Section~\ref{sec:nw}, in $156$
ethnic groups. Whatever the basis, the instrument has an unclustered first-stage $F$ near
$\num{1258}$ and a cluster-robust $F$ of $17.0$, and the naive 2SLS estimate with geography
entered linearly is $-0.243$ for trust in neighbours and $-0.257$ for trust in relatives.
Each row refits the flexible propensity in the stated basis in latitude and longitude. Here
$\hat\varrho$ is the contamination share with a $95\%$ cluster bootstrap interval
($B=999$), the corrected $\hat\beta$ the residualized propensity estimand, the corrected
$F$ its own cluster-robust first stage, and $p_X$ the bias-targeting test of
Section~\ref{sec:exact}.}
\label{tab:nwrobust}
\begin{tabular}{llcccc}
\toprule
outcome & basis & $\hat\varrho$ $[95\%]$ & corrected $\hat\beta$ & corrected $F$ & $p_X$\\
\midrule
trust in neighbours & quadratic & $0.52\,[0.08,0.92]$ & $-0.205$ & $2.7$ & $0.91$\\
                    & cubic     & $0.52\,[0.12,0.96]$ & $-0.190$ & $2.4$ & $0.98$\\
                    & spline    & $0.52\,[0.09,1.03]$ & $-0.200$ & $3.0$ & $0.98$\\
\addlinespace
trust in relatives  & quadratic & $0.51\,[0.08,0.92]$ & $-0.137$ & $2.7$ & $0.89$\\
                    & cubic     & $0.52\,[0.11,0.95]$ & $-0.127$ & $2.4$ & $0.74$\\
                    & spline    & $0.51\,[0.10,1.04]$ & $-0.170$ & $3.0$ & $0.95$\\
\bottomrule
\end{tabular}
\end{table}

\begin{table}[h]\centering
\caption{Contamination share. Population value, plug-in estimate at $n=3000$,
$95\%$ percentile-bootstrap coverage of the population share ($\num{1000}$
replications, $B=499$), and the induced 2SLS bias.}
\label{tab:share}
\begin{tabular}{rrrrr}
\toprule
$\kappa$ & $\varrho$ (pop.) & $\hat\varrho$ ($n{=}3000$) & $95\%$ coverage & bias\\
\midrule
0.25 & 0.021 & 0.022 & 0.949 & 1.306\\
0.50 & 0.082 & 0.080 & 0.953 & 2.549\\
1.00 & 0.239 & 0.240 & 0.949 & 4.365\\
2.00 & 0.494 & 0.497 & 0.950 & 6.340\\
4.00 & 0.682 & 0.683 & 0.933 & 7.539\\
8.00 & 0.735 & 0.741 & 0.942 & 7.659\\
16.0 & 0.752 & 0.754 & 0.940 & 7.814\\
\bottomrule
\end{tabular}
\end{table}

\begin{figure}[h]\centering
\includegraphics[width=\textwidth]{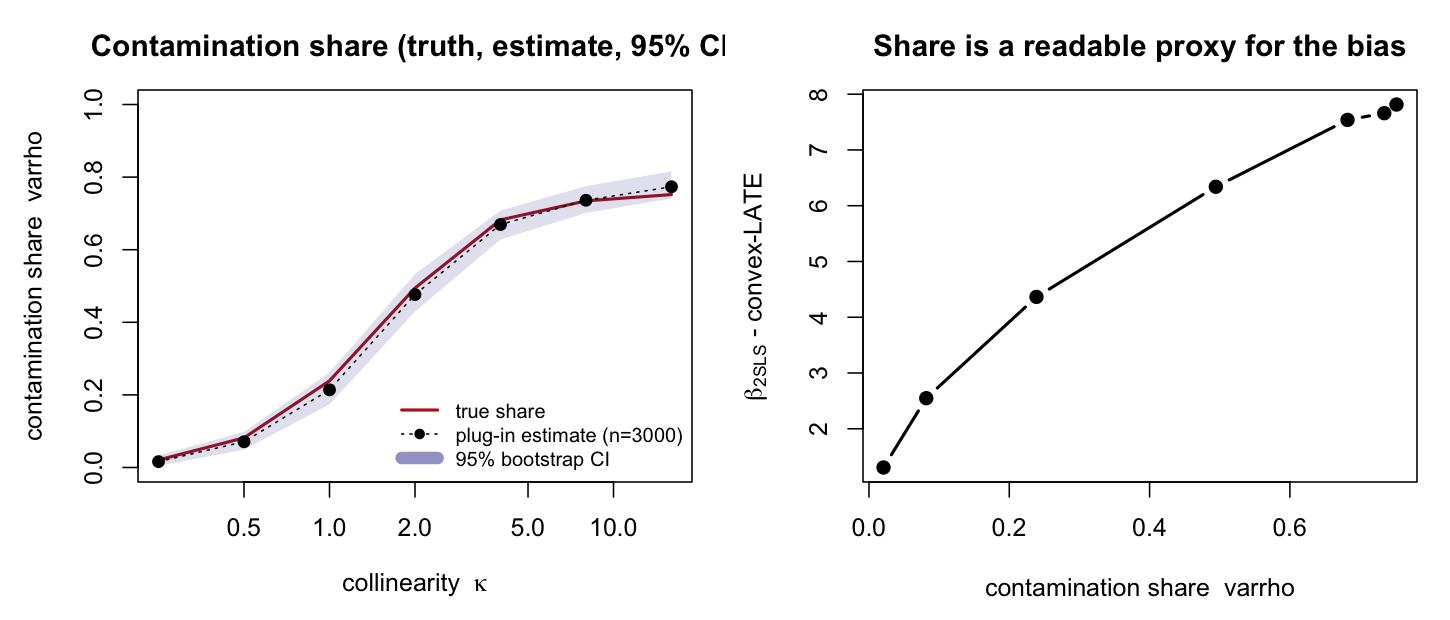}
\caption{Left panel. The population contamination share (solid), the plug-in
estimate at $n=3000$ (points) and its $95\%$ bootstrap band, against the
collinearity knob $\kappa$. Right panel. Along this one-parameter family the
population bias is monotone in the share, so here $\hat\varrho$ orders these designs
by danger. This ordering is a feature of the family swept, not a general property.
Across unrelated designs bias also depends on $\theta_Y$ and $\bar L$, so equal
shares can accompany different biases.}
\label{fig:share}
\end{figure}

\clearpage

\clearpage
\begin{thebibliography}{99}
%=============================================================================

\bibitem[Andrews et~al.(2019)]{andrews2019}
Andrews, I., Stock, J.H., Sun, L., 2019. Weak instruments in instrumental
variables regression: theory and practice. Annual Review of Economics 11,
727--753.

\bibitem[Angrist et~al.(1996)]{angrist1996}
Angrist, J.D., Imbens, G.W., Rubin, D.B., 1996. Identification of causal effects
using instrumental variables. Journal of the American Statistical Association 91,
444--455.

\bibitem[Abadie(2003)]{abadie2003}
Abadie, A., 2003. Semiparametric instrumental variable estimation of treatment
response models. Journal of Econometrics 113, 231--263.

\bibitem[Bierens(1990)]{bierens1990}
Bierens, H.J., 1990. A consistent conditional moment test of functional form.
Econometrica 58, 1443--1458.

\bibitem[Nunn and Wantchekon(2011)]{nunn2011}
Nunn, N., Wantchekon, L., 2011. The slave trade and the origins of mistrust in
Africa. American Economic Review 101, 3221--3252.

\bibitem[Olson(1998)]{olson1998}
Olson, C.A., 1998. A comparison of parametric and semiparametric estimates of the
effect of spousal health insurance coverage on weekly hours worked by wives. Journal
of Applied Econometrics 13, 543--565.

\bibitem[Poterba et~al.(1995)]{poterba1995}
Poterba, J.M., Venti, S.F., Wise, D.A., 1995. Do 401(k) contributions crowd out
other personal saving? Journal of Public Economics 58, 1--32.

\bibitem[Wooldridge(2010)]{wooldridge2010}
Wooldridge, J.M., 2010. Econometric Analysis of Cross Section and Panel Data, 2nd
ed. MIT Press.

\bibitem[Blandhol et~al.(2022)]{blandhol2022}
Blandhol, C., Bonney, J., Mogstad, M., Torgovitsky, A., 2022. When is TSLS
actually LATE? NBER Working Paper 29709.

\bibitem[Goldsmith-Pinkham et~al.(2024)]{goldsmith2024}
Goldsmith-Pinkham, P., Hull, P., Koles{\'a}r, M., 2024. Contamination bias in
linear regressions. American Economic Review 114, 4015--4051.

\bibitem[Zhao et~al.(2025)]{ddl2025}
Zhao, A., Ding, P., Li, F., 2025. Two-stage least squares with treatment--covariate
interactions for treatment effect heterogeneity. Working paper, arXiv:2502.00251.

\bibitem[Krumme and Westphal(2026)]{krumme2026}
Krumme, A., Westphal, M., 2026. Testing IV validity and LATE interpretation using
flexible covariate specifications. IZA Discussion Paper 18573.

\bibitem[Dossani et~al.(2024)]{dossani2024}
Dossani, A., Dotson, J.P., Schonlau, R., 2024. Contaminated control variables in
2SLS models. Working paper.

\bibitem[S{\l}oczy{\'n}ski et~al.(2026)]{ssu2026}
S{\l}oczy{\'n}ski, T., Sun, L., Uysal, S.D., 2026. A practical guide to
instrumental variables methods with heterogeneous treatment effects. Working
paper, arXiv:2605.15115.

\bibitem[Young(2024)]{young2024}
Young, A., 2024. Nearly collinear robust procedures for 2SLS estimation. The Stata
Journal 24, 3--28.

\bibitem[van der Vaart(1998)]{vandervaart1998}
van der Vaart, A.W., 1998. Asymptotic Statistics. Cambridge University Press.

\bibitem[Bound et~al.(1995)]{bound1995}
Bound, J., Jaeger, D.A., Baker, R.M., 1995. Problems with instrumental variables
estimation when the correlation between the instruments and the endogenous
explanatory variable is weak. Journal of the American Statistical Association 90,
443--450.

\bibitem[Chernozhukov et~al.(2018)]{chernozhukov2018}
Chernozhukov, V., Chetverikov, D., Demirer, M., Duflo, E., Hansen, C., Newey, W.,
Robins, J., 2018. Double/debiased machine learning for treatment and structural
parameters. The Econometrics Journal 21, C1--C68.

\bibitem[Conley et~al.(2012)]{conley2012}
Conley, T.G., Hansen, C.B., Rossi, P.E., 2012. Plausibly exogenous. Review of
Economics and Statistics 94, 260--272.

\bibitem[de Chaisemartin(2017)]{dechaisemartin2017}
de Chaisemartin, C., 2017. Tolerating defiance? Local average treatment effects
without monotonicity. Quantitative Economics 8, 367--396.

\bibitem[Heckman and Vytlacil(2005)]{heckman2005}
Heckman, J.J., Vytlacil, E., 2005. Structural equations, treatment effects, and
econometric policy evaluation. Econometrica 73, 669--738.

\bibitem[Imbens and Angrist(1994)]{imbens1994}
Imbens, G.W., Angrist, J.D., 1994. Identification and estimation of local average
treatment effects. Econometrica 62, 467--475.

\bibitem[van Kippersluis and Rietveld(2018)]{kippersluis2018}
van Kippersluis, H., Rietveld, C.A., 2018. Beyond plausibly exogenous. The
Econometrics Journal 21, 316--331.

\bibitem[Mogstad and Torgovitsky(2018)]{mogstad2018}
Mogstad, M., Torgovitsky, A., 2018. Identification and extrapolation of causal
effects with instrumental variables. Annual Review of Economics 10, 577--613.

\bibitem[Mogstad and Torgovitsky(2024)]{mogstad2024}
Mogstad, M., Torgovitsky, A., 2024. Instrumental variables with unobserved
heterogeneity in treatment effects, in: Handbook of Labor Economics. Elsevier.

\bibitem[Montiel Olea and Pflueger(2013)]{montielolea2013}
Montiel Olea, J.L., Pflueger, C., 2013. A robust test for weak instruments.
Journal of Business \& Economic Statistics 31, 358--369.

\bibitem[Newey(1985)]{newey1985}
Newey, W.K., 1985. Maximum likelihood specification testing and conditional
moment tests. Econometrica 53, 1047--1070.

\bibitem[Newey and McFadden(1994)]{newey1994}
Newey, W.K., McFadden, D., 1994. Large sample estimation and hypothesis testing,
in: Handbook of Econometrics, vol.~4. Elsevier, pp.~2111--2245.

\bibitem[Ramsey(1969)]{ramsey1969}
Ramsey, J.B., 1969. Tests for specification errors in classical linear
least-squares regression analysis. Journal of the Royal Statistical Society B 31,
350--371.

\bibitem[Sanderson and Windmeijer(2016)]{sanderson2016}
Sanderson, E., Windmeijer, F., 2016. A weak instrument $F$-test in linear IV
models with multiple endogenous variables. Journal of Econometrics 190, 212--221.

\bibitem[S{\l}oczy{\'n}ski(2022)]{sloczynski2022}
S{\l}oczy{\'n}ski, T., 2022. When should we (not) interpret linear IV estimands
as LATE? Working paper, Brandeis University.

\bibitem[Staiger and Stock(1997)]{staiger1997}
Staiger, D., Stock, J.H., 1997. Instrumental variables regression with weak
instruments. Econometrica 65, 557--586.

\bibitem[Stevenson(2026)]{stevenson2026}
Stevenson, M.T., 2026. Humans in the loop: the next frontier in the credibility revolution. Journal of Economic Literature, forthcoming.

\bibitem[Stock and Yogo(2005)]{stock2005}
Stock, J.H., Yogo, M., 2005. Testing for weak instruments in linear IV
regression, in: Identification and Inference for Econometric Models. Cambridge
University Press, pp.~80--108.

\bibitem[Tauchen(1985)]{tauchen1985}
Tauchen, G., 1985. Diagnostic testing and evaluation of maximum likelihood
models. Journal of Econometrics 30, 415--443.

\end{thebibliography}
\end{document}